\documentclass[conference]{IEEEtran}
\IEEEoverridecommandlockouts

\usepackage{amsmath,amsthm}
\usepackage{stmaryrd}
\usepackage{tikz}
\usetikzlibrary{arrows, automata, positioning,calc}
\usetikzlibrary{decorations.pathmorphing}
\usetikzlibrary{decorations.pathreplacing}
\usepackage{comment}
\usepackage{color}
\usepackage{gastex}
\usepackage{algorithmic}
\usepackage{algorithm}
\usepackage{amssymb}
\usepackage{mathabx} 
\usepackage{kbordermatrix}

\newcommand{\RP}{\mathbb{R}_{\ge 0}}
\newcommand{\R}{\mathbb{R}}
\newcommand{\N}{\mathbb{N}}

\newcommand{\defequals}{\stackrel{\mathrm{def}}{=}}

\newcommand{\A}{\mathcal{A}} 
\newcommand{\locs}{L} 
\newcommand{\loc}{\ell} 
\newcommand{\clocks}{\mathcal{X}} 

\newcommand{\guard}{\varphi} 
\newcommand{\reset}{\lambda}
\newcommand{\val}{\nu} 

\newcommand{\prop}{AP}
\newcommand{\lbl}{LB}

\makeatletter
\newcommand*{\defeq}{\mathrel{\rlap{%
                     \raisebox{0.3ex}{$\m@th\cdot$}}%
                     \raisebox{-0.3ex}{$\m@th\cdot$}}%
                     =}
\makeatother

\makeatletter
\newcommand*{\ndefeq}{\mathrel{\rlap{%
                     \raisebox{0.3ex}{$\m@th\cdot$}}%
                     \raisebox{-0.3ex}{$\m@th\cdot$}%
                     \rlap{%
                     \raisebox{0.3ex}{$\m@th\cdot$}}
                     \raisebox{-0.3ex}{$\m@th\cdot$}}
                     =}
\makeatother

\newcommand{\tctl}{{TCTL}} 
\newcommand{\ptctl}{{PTRL}} 

\newcommand{\F}{\lozenge} 
\newcommand{\G}{\boxempty} 
\newcommand{\true}{\texttt{true}} 

\newcommand{\tuple}[1]{\langle #1 \rangle}
\def\abs#1{\ensuremath{\lvert #1\rvert}}

\newcommand{\sto}[1]{\stackrel{#1}{\Rightarrow}}
\newcommand{\params}{\Theta} 
\newcommand{\param}{\theta} 
\newcommand{\zeroval}{\boldsymbol{0}} 
\newcommand{\vect}[1]{\left( \begin{smallmatrix} #1 \end{smallmatrix} \right) }
\newcommand{\zone}{\mathit{Z}}

\newcommand{\semantics}[1]{\llbracket #1\rrbracket}

\newcommand{\M}{\mathcal{M}} 
\newcommand{\closure}{\mathit{closure}}

\newcommand{\edges}{E}
\newcommand{\fract}{\mathrm{frac}}
\newcommand{\floor}[1]{ \lfloor{ #1 }\rfloor}

\newcommand{\rrel}[2]{\mathrel{\mathcal{R}^{#1}_{#2}}}
\newcommand{\zrel}[2]{\mathrel{\mathcal{Z}^{#1}_{#2}}}

\newtheorem{theorem}{Theorem}
\newtheorem{lemma}[theorem]{Lemma}
\newtheorem{corollary}[theorem]{Corollary}
\newtheorem{proposition}[theorem]{Proposition}

\newtheorem{example}{Example}

\title{Revisiting Reachability in Timed Automata}

\author{
\IEEEauthorblockN{Karin Quaas\IEEEauthorrefmark{1},
Mahsa Shirmohammadi\IEEEauthorrefmark{2},
James Worrell\IEEEauthorrefmark{2}}
\IEEEauthorblockA{\IEEEauthorrefmark{1}Universit\"at Leipzig, Germany}
\IEEEauthorblockA{\IEEEauthorrefmark{2}University of Oxford, UK}
}

\begin{document}

\maketitle
\begin{abstract}
We revisit a fundamental result in real-time verification, namely
that the binary reachability relation between configurations of a
given timed automaton is definable in linear arithmetic over the
integers and reals.  In this paper we give a new and simpler proof of this result,
building on the well-known reachability analysis of timed automata
involving difference bound matrices.  Using this new proof, we give an
exponential-space procedure for model checking the reachability
fragment of the logic parametric TCTL.  Finally we show that the
latter problem is NEXPTIME-hard.

\end{abstract}


\section{Introduction} 
The PSPACE-completeness of the reachability problem for timed automata
is arguably the most fundamental result in real-time verification.
This theorem was established by Alur and Dill in 
paper~\cite{AlurD94} for which they were awarded the Alonzo Church
award in 2016.  The reachability problem has been intensively studied
in the intervening 20 years, leading to practical algorithms and
generalisations to more expressive models.  As of now, \cite{AlurD94}
is the most cited paper that has appeared in the journal
\emph{Theoretical Computer Science}.

Properly speaking, Alur and Dill considered reachability
  between \emph{control states} (also called \emph{locations}).  The
  problem of computing the binary reachability relation
  over \emph{configurations} (both control states and clock
  valuations) is more involved.  Here the main result is due to Comon
  and Jurski~\cite{ComonJ99}, who showed that the reachability
  relation of a given timed automaton is effectively definable by a
  formula of first-order linear arithmetic over the reals augmented
  with a unary predicate denoting the integers.  Importantly, this
  fragment of mixed linear arithmetic has a decidable satisfiability
  problem, e.g., by translation to S1S.

  Despite its evident utility, particularly for parametric
  verification, it is fair to say that the result of Comon and Jurski
  has proven less influential than that of Alur and Dill.  We believe
  that this is due both to the considerable technical complexity of
  the proof, which runs to over 40 pages in~\cite{ComonJ99-TR}, as
  well as the implicit nature of their algorithm, making it hard to
  extract complexity bounds.

In this paper we revisit the result of Comon and Jurski.  Our
two main contributions as follows:
\begin{itemize}
\item We give a new and conceptually simpler proof that generalises the
  classical reachability algorithm for timed automata involving
difference bound matrices and standard operations thereon.  The key
new idea is to carry out the algorithm on a symbolically presented
initial configuration.  This approach is fundamentally different from
that of~\cite{ComonJ99}, the main part of which involves a syntactic
transformation showing that every timed automaton can be effectively
emulated by a \emph{flat} timed automaton, i.e., one that does not
contain nested loops in its control graph.


\item We apply our strengthened formulation of the Comon-Jurski result
  to parametric model checking.  We show that the formula representing
  the reachability relation can be computed in time singly exponential
  in the size of the timed automaton.  Using this bound on the formula
  size and utilising results of~\cite{FerranteR75,To09} on
  quantifier-elimination for first-order logic over the reals and
  integers, we show that the model checking problem for the
  reachability fragment of the temporal logic \emph{parametric TCTL}
  is decidable in exponential space.  We show in the main body of the
  paper that this problem is NEXPTIME-hard and sketch in the
  conclusion how to obtain matching upper and lower bounds.
\end{itemize}

There are two main steps in our approach to computing a formula
representing the reachability relation.  First, given a timed
automaton $\mathcal{A}$ and a configuration $\langle \ell,\nu \rangle$
of $\A$, we construct a version of the region automaton
of~\cite{AlurD94} that represents all configurations reachable from
$\langle \ell,\nu \rangle$.  Unlike~\cite{AlurD94} we do not identify
all clock values above the maximum clock constant; so our version of
the region automaton is a counter machine rather than a finite state
automaton.  The counters are used to store the integer parts of clock
valuations of reachable configurations, while the fractional parts of
the clock valuations are aggregated into zones that are represented
within the control states of the counter machine by difference bound
matrices.  Since the counters mimic clocks they are monotonic and 
so the reachability relation on such a counter machine is definable
in a weak fragment of Presburger arithmetic.

The second step of our approach is to make the previous construction
parametric: we show that the form of the counter machine does not
depend on the precise numerical values of the clocks in the initial
valuation $\nu$, just on a suitable logical \emph{type} of $\nu$.
Given such a type, we develop a parametric version of the
counter-machine construction.  Combining this construction with the
fact that the reachability relation for the considered class of
counter machines is definable in a fragment of Presburger arithmetic,
we obtain a formula that represents the full reachability relation of
the timed automaton $\mathcal{A}$.

\subsection{Related Work}
Dang~\cite{Dang03} has generalised the result of Comon and Jurski,
showing that the binary reachability relation for pushdown timed
automata is definable in linear arithmetic.  The approach
in~\cite{Dang03} relies on a finite partition of the fractional parts
of clock valuations into so-called \emph{patterns}, which play a role
analogous to types in our approach.  The notion of pattern is ad-hoc
and, as remarked by Dang, relatively complicated.  In particular,
patterns lack the simple characterisation in terms of difference
constraints that is possessed by types.  The latter is key to our
result that the reachability relation can be expressed by a Boolean
combination of difference constraints.

Dima~\cite{Dima02} gives an automata theoretic representation of the
reachability relation of a timed automaton.  To this end he introduces
a class of automata whose runs encode tuples in such a relation.  The
main technical result of~\cite{Dima02} is to show that this class of
automata is effectively closed under relational reflexive-transitive
closure.

The model checking problem for parametric TCTL was studied by
Bruy\`{e}re~\emph{et al.}~\cite{BruyereDR08,BruyereDR03} in the case of
integer-valued parameters.  Here we allow real-valued parameters,
which leads to a strictly more expressive semantics.

Parametric DBMs have been used in~\cite{AnnichiniAB00,HuneRSV02} to
analyse reachability in parametric timed automata.  These are related
to but different from the parametric DBMs occurring in
Subsection~\ref{symbolic}.

\subsection{Organisation}
We introduce and state our main results in the body of the paper.  The
central constructions underlying our proofs are also given in the
body, along with illustrative examples. 
Many of the proof details are
relegated to the appendix.


\section{Main Definitions and  Results}
\subsection{Timed Automata}
\label{sec:TA}
Given a set~$\clocks=\{x_1,\ldots,x_n\}$ of \emph{clocks}, the set
$\Phi(\clocks)$ of \emph{clock constraints} is generated by the
grammar
\[ \varphi ::= \true \mid x<k \,\mid\, x = k \,\mid\, x>k \,\mid\,
  \varphi \wedge \varphi \, , \] where $k \in \N$ is a natural number
  and $x\in \clocks$.  A \emph{clock valuation} is a mapping~$\val:
  \clocks \to \RP$, where $\RP$ is the set of non-negative real
  numbers.  We denote by $\zeroval$ the valuation such
  that~$\zeroval(x)=0$ for all $x\in \clocks$.  Let~$\RP^{\clocks}$ be
  the set of all clock valuations.  We write $\val\models\varphi$ to
  denote that~$\val$ satisfies the constraint $\guard$.  Given
  $t\in\RP$, we let $\val+t$ be the clock valuation such that
  $(\val+t)(x)=\val(x)+t$ for all clocks~$x\in\clocks$.  Given
  $\reset\subseteq\clocks$, let $\val[\reset\leftarrow 0]$ be the
  clock valuation such that $\val[\reset\leftarrow 0](x)=0$ if
  $x\in\reset$, and $\val[\reset\leftarrow 0](x)=\val(x)$ if
  $x\not\in\reset$.  We typically write $\val_i$ as shorthand for
  $\val(x_i)$, and by convention we define $\val_0=0$.  For all~$r\in
  \mathbb{R}$, let $\fract(r)$ be the fractional part of~$r$, and
  $\floor{r}$ be the integer part.  Denote by~$\fract(\val)$ and
  $\floor{\val}$ the valuations such that
  $(\fract(\val))(x_i)=\fract(\val_i)$ and
  $\floor{\val}(x_i)=\floor{\val_i}$ for all clocks~$x_i\in \clocks$.

A \emph{timed automaton} is a tuple $\A=\tuple{\locs,\clocks,\edges}$,
where~$\locs$ is a finite set of \emph{locations}, $\clocks$ is a
finite set of \emph{clocks} 
and $\edges\subseteq \locs\times \Phi(\clocks)\times 2^\clocks\times
\locs$ is the set of \emph{edges}.

The semantics of a timed automaton~$\A=\tuple{\locs,\clocks,\edges}$
is given by a labelled transition system~$\tuple{Q,\sto{}}$ with set
of \emph{configurations} $Q=\locs \times \RP^{\clocks}$ and set of
\emph{transition labels}~$\RP$.  A configuration~$\tuple{\loc,\val}$
consists of a location~$\loc$ and a clock valuation~$\val$.  Given two
configurations~$\tuple{\loc,\val}$ and $\tuple{\loc',\val'}$, we
postulate:
\begin{itemize}
\item a delay transition~$\tuple{\ell,\val}\sto{d}\tuple{\ell',\val'}$ for
    some $d\geq 0$, if~$\val'=\val+d$ and $\ell=\ell'$;
	\item a discrete transition~$\tuple{\ell,\val}\sto{0}\tuple{\ell',\val'}$, if there 
is an edge $\tuple{\ell,\varphi,\reset,\ell'}$ of $\A$
	such that $\val\models\guard$ and $\val'=\val[\reset \leftarrow 0]$. 
\end{itemize}

A \emph{run}~$\rho=q_0 \sto{d_1} q_1 \sto{d_2} q_2 \sto{d_3} \ldots$
of $\A$ is a (finite or infinite) sequence of delay and discrete
transitions in \mbox{$\tuple{Q,\sto{}}$}.  We require infinite runs to
have infinitely many discrete transitions and to be \emph{non-zeno},
that is, we require $\sum_{i=1}^\infty d_i$ to diverge.

Henceforth we assume that in any given timed automaton with
set~$\clocks$ of clocks, $x_n$ is a special reference clock that is never reset.
Clearly this assumption is without loss of generality for encoding the
reachability relation.

Note that we consider timed automata without \emph{diagonal
constraints}, that is, guards of the form $x_i-x_j \sim k$, for $k$ an
integer.  It is known that such constraints can be removed without
affecting the reachability relation (see~\cite{AlurD94,BerardPDG98}).

\subsection{Linear Arithmetic}\label{sec:Logic}
In this section we introduce a first-order
language~$\mathcal{L}_{\mathbb{R},\mathbb{Z}}$ in which to express the
reachability relation of a timed automaton.

Language $\mathcal{L}_{\mathbb{R},\mathbb{Z}}$ has two sorts: a
real-number sort and an integer sort.  The collection
$\mathcal{T}_{\mathbb{R}}$ of terms of real-number sort is specified by
the grammar
\[ t::= c \, \mid \, r \mid \, t+t \, \mid t-t \, , \] where
$c\in \mathbb{Q}$ is a constant and $r\in\{r_0,r_1,\ldots\}$ is a
real-valued variable.  Given terms
$t,t'\in \mathcal{T}_{\mathbb{R}}$, we have an atomic formula
$t \leq t'$.  The collection $\mathcal{T}_{\mathbb{Z}}$ of terms of
integer sort is specified by the grammar
\[ t::= c \, \mid \, z \, \mid \, t+t \, \mid t -t \, , \] where
$c\in \mathbb{Z}$ is a constant and $z\in\{z_0,z_1,\ldots\}$ is an
integer variable.  Given terms $t,t'\in\mathcal{T}_{\mathbb{Z}}$, we
have atomic formulas $t \leq t'$ and ${ t\equiv t'} \pmod m$, where
$m\in \mathbb{Z}$.  Formulas of $\mathcal{L}_{\mathbb{R},\mathbb{Z}}$
are constructed from atomic formulas using Boolean connectives and
first-order quantifiers.  

Throughout the paper we consider a fixed semantics for
$\mathcal{L}_{\mathbb{R},\mathbb{Z}}$ over the two-sorted structure in
which the real-number sort is interpreted by $\mathbb{R}$, the integer
sort by $\mathbb{Z}$, and with the natural interpretation of addition
and order on each sort.

The sublanguage $\mathcal{L}_{\mathbb{R}}$ of
$\mathcal{L}_{\mathbb{R},\mathbb{Z}}$ involving only terms of
real-number sort is called \emph{real arithmetic}.  The sublanguage
$\mathcal{L}_{\mathbb{Z}}$ involving only terms of integer sort is
called \emph{Presburger arithmetic}.  Optimal complexity bounds for
deciding satisfiability of sentences of real arithmetic and Presburger
arithmetic are given in~\cite{Berman80} with, roughly speaking, real
arithmetic requiring single exponential space and Presburger
arithmetic double exponential space.

\begin{proposition}
 Deciding the truth of a sentence in the existential 
fragment of $\mathcal{L}_{\mathbb{R},\mathbb{Z}}$ can be done in NP.
\end{proposition}
\begin{proof}
The respective decision problems
for the existential fragment of real arithmetic and the existential 
fragment of  Presburger arithmetic are in NP~\cite{Sontag85}, \cite{Weispfenning88}. 
Deciding the truth of a sentence in the existential 
fragment of $\mathcal{L}_{\mathbb{R},\mathbb{Z}}$
 is therefore also in NP, since we can guess truth
values for the Presburger and real-arithmetic subformulas, and
separately check realisability of the guessed truth values in
non-deterministic polynomial time.
\end{proof}

For the purpose of model checking, it will be
useful to establish complexity bounds for a language
$\mathcal{L}^*_{\mathbb{R},\mathbb{Z}}$, intermediate
between $\mathcal{L}_{\mathbb{R}}$ and the full language
$\mathcal{L}_{\mathbb{R},\mathbb{Z}}$.  The language
$\mathcal{L}^*_{\mathbb{R},\mathbb{Z}}$ arises from
$\mathcal{L}_{\mathbb{R},\mathbb{Z}}$ by restricting the atomic
formulas over terms of integer sort to have the form
\begin{gather}
 z-z' \leq c \, \mid \, z \leq c \, \mid \,
   {z-z' \equiv c} \pmod d 
\label{eq:diff1}
\end{gather}
for integer variables $z,z'$ and integers $c,d$.

\begin{proposition}
  Deciding the truth of a prenex-form sentence
  $Q_1x_1\ldots Q_nx_n \, \varphi$ in
  $\mathcal{L}^*_{\mathbb{R},\mathbb{Z}}$ can be done in
  space exponential in $n$ and polynomial in $\varphi$.
\label{prop:diff-bound}
\end{proposition}
\begin{proof}
The proposition is known to hold separately for
$\mathcal{L}_{\mathbb{R}}$~\cite{FerranteR75} and for the fragment of
$\mathcal{L}_{\mathbb{Z}}$ in which atomic formulas have the form
shown in (\ref{eq:diff1})~\cite[Section 4]{To09}.  The respective
arguments of~\cite{FerranteR75} and~\cite{To09} can be
straightforwardly combined to prove the proposition; see
Section~\ref{append-prop-diff-bound} for details.
\end{proof}

\subsection{Definability of the Reachability Relation}
Given a timed automaton $\A$ with $n$ clock variables, we express the
reachability relation between every pair of locations~$\ell,\ell'$ by
a formula
\[ \varphi_{\ell,\ell'}(z_1,\ldots,z_n,r_1,,
  \ldots,r_n,z'_1,\ldots,z'_n,r'_1,,\ldots,r'_n), \]
in the existential fragment of $\mathcal{L}_{\mathbb{R},\mathbb{Z}}$ where
$z_1,z'_1,\ldots,z_n,z'_n$ are integer variables and
$r_1,r'_1,\ldots,r_n,r'_n$ are real variables ranging over the
interval~$[0,1]$.  Our main result,
Theorem~\ref{theorem_main_ta_formula}, shows that there is a finite
run in $\A$ from configuration~$\tuple{ \ell,\val}$ to
configuration~$\tuple{ \ell',\val'}$ just in case
\[\begin{aligned} 
\langle &\floor{\val_1},\ldots,
 \floor{\val_n},\fract(\val_1), \ldots, \fract(\val_n),\\
&\floor{\val'_1},\ldots,
 \floor{\val'_n},\fract(\val'_1),\cdots, \fract(\val'_n) 
\rangle \models \varphi_{\ell,\ell'}.
\end{aligned}\]

\begin{example} Consider the following timed automaton:
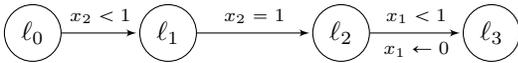
\begin{figure}[h!]
\begin{center}
\begin{tikzpicture}[>=latex',shorten >=1pt,node distance=1.9cm,on grid,auto,
roundnode/.style={circle, draw,minimum size=1.2mm}]

\node [roundnode] (li) at(0,0) {$\ell_0$};
\node [roundnode] (l0) [right=1.8cm of li] {$\ell_1$};
\node [roundnode](l1)  [right=2.3cm of l0] {$\ell_2$};
\node [roundnode] (l2)  [right=2cm of l1] {$\ell_3$};

\path[->] (li) edge  node [above,midway] {\scriptsize{$x_2<1$}}  (l0);
\path[->] (l0) edge  node [above,midway] {\scriptsize{$x_2=1$}} node[below, midway]{ } (l1);
\path[->] (l1) edge  node [above,midway] {\scriptsize{$x_1<1$}} node[below, midway]{ \scriptsize{$x_1\leftarrow 0$}} (l2);

\end{tikzpicture}
\end{center}
\vspace{-.4cm}
\end{figure}

A brief inspection reveals that location~$\loc_3$ can be reached
from a configuration~$\tuple{\loc_0, \vect{\val_1\\\val_2}}$ if and
only if $\val_1<\val_2< 1$.  The reachability relation between
locations $\ell_0$ and $\ell_3$ is expressed by the formula
\[\begin{aligned} 
	 \varphi_{\ell_0,\ell_3}(z_1,z_2,&r_1,r_2,z'_1,z'_2,r'_1,r'_2)  	\, \defequals  (z_1=z_2=0)\\
	&  \wedge(r_1<r_2<1) \\
\, & \wedge ((z'_2-z'_1=1 \wedge 0\leq r'_2-r'_1< r_2-r_1) \\
	& \vee (z'_2-z'_1=2 \wedge 0 \leq 1+r'_2-r'_1<r_2-r_1)),
\end{aligned}\]
where the real-valued variables~$r_1,r_2,r'_1,r'_2$   range over the interval~$[0,1]$.
\end{example}

\begin{example} Consider the following timed automaton:
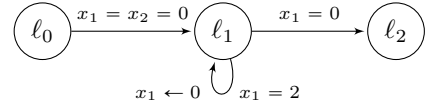
\begin{figure}[h!]
\begin{center}
\begin{tikzpicture}[>=latex',shorten >=1pt,node distance=1.9cm,on grid,auto,
roundnode/.style={circle, draw,minimum size=1.2mm}]

\node [roundnode] (l0) at(0,0) {$\ell_0$};
\node [roundnode](l1)  [right=2.4cm of l0] {$\ell_1$};
\node [roundnode](l2)  [right=2.3cm of l1] {$\ell_2$};

\path[->] (l1) edge[loop below]  node [midway,right] {\scriptsize{$~x_1=2$}} node [left, midway]{\scriptsize{$x_1\leftarrow 0~~$}} (l1);
\path[->] (l0) edge  node [above,midway] {\scriptsize{$x_1=x_2=0$}}  (l1);
\path[->] (l1) edge  node [above,midway] {\scriptsize{$x_1=0$}}  (l2);
\end{tikzpicture}
\vspace{-.4cm}
\end{center}
\end{figure}

\noindent We have
\[\begin{aligned} 
 \varphi_{\ell_0,\ell_3}(z_1,z_2,r_1,r_2,&z'_1,z'_2,r'_1,r'_2)  \, \defequals \\
 & (r_1=r_2=0)\wedge (r'_1= r'_2)  \wedge\\\,
&(z_1=z_2=0) \wedge ({z'_2-z_1' \equiv 0} \;(\bmod\; 2)). \\
\end{aligned}\]
\end{example}

\subsection{Parametric Timed Reachability Logic}

\begin{figure*}[t]
\begin{minipage}{.6\linewidth} 

 \begin{tikzpicture}[>=latex',shorten >=1pt,node distance=1.9cm,on grid,auto,
roundnode/.style={circle, draw,minimum size=1.2mm},
squarenode/.style={rectangle, draw,dashed,minimum size=2mm}]

\draw (0,0) node [text=white,fill=white] {dummy};
\node [roundnode] (l0) at(.5,0) {$\ell_0$};

\node [roundnode](l1)  [right=2.3cm of l0] {$\ell_1$};
\node [squarenode](lll1) [above=.7cm of l1] {$p_1$};
\node [roundnode] (l2)  [right=2cm of l1] {$\ell_2$};
\node [roundnode] (l3)  [right=2cm of l2] {$\ell_3$};
\node [squarenode](lll3) [above=.7cm of l3] {$p_2$};
\node [roundnode](l4)  [right=2cm of l3] {$\ell_4$};

\path[->] (l0) edge  node [above,midway] {\scriptsize{$0<x_1<1$}} node[below, midway]{ \scriptsize{$x_1\leftarrow 0$}} (l1);
\path[->] (l1) edge  node [above,midway] {\scriptsize{$x_1=0$}}  (l2);
\path[->] (l2) edge  node [above,midway] {\scriptsize{$x_2=1$}}  (l3);
\path[->] (l3) edge  node [above,midway] {\scriptsize{$x_2=1$}} node[below, midway] { \scriptsize{$x_1\leftarrow 0$}} (l4);

\path[->] (l4) edge [loop above] node [above, midway]{\scriptsize{$x_1\leftarrow 0$}} (l4);

\end{tikzpicture}

\caption{
	A timed automaton where the satisfaction relation of {\ptctl} with parameters ranging over non-negative real numbers is 
	different from the relation when parameters are restricted to naturals. The locations~$\ell_1$ and~$\ell_3$ are labelled by propositions~$p_1$ 
and~$p_2$, respectively. 
	The set $\reset$ of clocks that are reset by a transitions are shown by~$\reset \leftarrow 0$; for example, 
  the transition from  $\loc_3$ to $\loc_4$ is guarded by~$x_2=1$ and resets~$x_1$.
	For all $0<\param<1$, we have $(\loc_0,\zeroval) \models \exists \, \F (p_1 \wedge \exists \, \F_{=\param} \, p_2)$, whereas there exists no~$n\in \N$ such that
	$(\loc_0,\zeroval) \models \exists \, \F (p_1 \wedge \exists \, \F_{=n}\, p_2)$.
	}
\label{fig:realParameters}
  \end{minipage}
\quad
\begin{minipage}{.35\linewidth}     
			\begin{center}
          \begin{tikzpicture}
    \draw[very thin,color=gray,step=.25cm] (0,0) grid (1.25,1.25);    
     \draw[->,thick] (0,0) -- (1.35,0) node[right] {$x_1$};
    \draw[->,thick] (0,0) -- (0,1.35) node[left] {$x_2$};

\draw[-,line width=2pt,blue]  (.75,0)--(1.25,.5);

\end{tikzpicture}
					$\kbordermatrix{
					&  x_0&  x_1&   x_2\\
					x_0 &  (\leq,0)&  (\leq,-0.6)&  (\leq,0)\\
					x_1& (\leq, 1)&     (\leq,0)&   (\leq,0.6)\\
					x_2& (\leq,0.4)&  (\leq,-0.6)&     (\leq,0)}
					$
     \caption{A DBM~$M$ with a zone~$\zone=\semantics{M}$.}\label{fig:d0}
\end{center}
  \end{minipage}
	\end{figure*}

 Timed computation tree logic (\tctl) is an
        extension of computation tree logic for specifying real-time
        properties~\cite{DBLP:journals/iandc/AlurCD93}.
        In~\cite{BruyereDR08} {\tctl} was generalised to allow
        parameters within timing constraints, yielding the logic
        \emph{parametric {\tctl}}.  In this paper we consider the
        fragment of parametric TCTL generated by the reachability
        modality $\exists\F$, which we call \emph{parametric timed
          reachability logic (\ptctl)}.

Let $\prop$ be a set of atomic propositions and $\params$ a set of
parameters.  Formulas of {\ptctl} of the \emph{first type} are given
by the grammar
\begin{gather} \varphi ::= p \,\mid\, \varphi\wedge\varphi \,\mid\, 
\neg\varphi  \,\mid\, \exists \F_{\sim \alpha} \,\varphi \, ,
\label{eq:grammar}
\end{gather}
where $p\in\prop$, ${\sim} \in\{<,\le, =,\ge,>\}$, and
$\alpha \in \mathbb{Q}\cup \params$. Formulas of {\ptctl}  
of the \emph{second type} are given by grammar
\begin{gather}
\psi ::= \varphi \,\mid\, \param-\param' \sim c \,\mid\, 
\psi_1 \wedge \psi_2 \,\mid\, \neg \psi \,\mid\, \exists \param \, \psi \, ,
\end{gather}
where $\varphi$ is a formula of the first type,
$\param,\param'\in\params$, ${\sim} \in\{<,\le, =,\ge,>\}$, and $c \in
\mathbb{Q}$.  In the sequel we use $\forall\G_{\sim \alpha} \varphi$
as abbreviation for $\neg\exists\F_{\sim \alpha}\neg\varphi$.

Formulas of {\ptctl} are interpreted with respect to a timed
automaton~$\A=\tuple{\locs,\clocks,\edges}$ and 
\emph{labelling function}~$\lbl:\locs\rightarrow 2^{\prop}$.  A
\emph{parameter valuation} is a function $\xi:\params \rightarrow \RP$.
Such a function is extended to the rational numbers by writing
$\xi(c)=c$ for $c \in \mathbb{Q}$.  Given a parameter valuation $\xi$, we
define a satisfaction relation $\models_{\xi}$ between configurations of
$\A$ and {\ptctl} formulas by induction over the structure of formulas.
The Boolean connectives are handled in the expected way, and we define

\begin{itemize}
\item[] $q \models_{\xi} \param - \param' \sim c$ iff 
$\xi(\param)-\xi(\param') \sim c$.
\item[] $q \models_{\xi} \exists \F_{\sim \alpha}\,\varphi$ iff  
there exists some infinite non-zeno run~$\rho=q_0 \sto{d_1}
q_1 \sto{d_2} q_2 \sto{d_3} \ldots$ of~$\A$ and $i\in \N$ such that
$q_0=q$, $d_1+\ldots+d_{i} \sim \xi(\alpha)$, and
$q_i \models \varphi$.
\item[] $q \models_{\xi} \exists \param \,\psi$ iff there exists a parameter
  valuation $\xi'$ such that $q \models_{\xi'} \psi$ and $\xi,\xi'$ agree on
  $\params \setminus \{ \param \}$.
\end{itemize}

\begin{example}
The $\ptctl$-formula $\forall\param (\exists\F_{<\theta} p_1 \rightarrow
\exists\F_{<\theta} p_2)$ expresses that some $p_2$-state is reachable in at most
the same time as any $p_1$-state is reachable.
\end{example}

The paper~\cite{BruyereDR08} considered a semantics for parametric
{\tctl} in which parameters range over naturals~$\N$.  Here we have
given a more general semantics in which parameters range over
non-negative real numbers~$\RP$.  The following example shows that the
satisfaction relation changes under this extension.

\begin{example}\label{examreal}
Consider the timed automaton in Figure~\ref{fig:realParameters} with
two clocks~$x_1,x_2$.  Clock valuations~$\val$ are denoted
by vectors~$\vect{\val_1\\ \val_2}$.  Let $\varphi =\exists  \F
(p_1 \wedge \exists  \F_{=\param} \,p_2)$.
All non-zeno infinite runs of the timed automaton, from
configuration~$\tuple{\loc_0,\zeroval}$, start with the following
prefix
\begin{multline*}
(\loc_0,\vect{0\\ 0})\sto{t}(\loc_1, \vect{0\\ t})\sto{0}(\loc_2, \vect{0\\ t})\\\sto{1-t}(\loc_3, \vect{1-t\\ 1})
\sto{0}(\loc_4, \vect{0\\ 1})
\end{multline*}
where $0<t<1$. Now we have that
$(\loc_1, \vect{0\\ t})\models (p_1 \wedge \exists  \F_{=1-t}\,
p_2)$.  As a result, $(\loc_0,\zeroval) \models \exists  \F
(p_1 \wedge \forall  \F_{=\param} \,p_2)$ only for $0< \param<1$.
Thus $(\loc_0,\zeroval) \models \mathop{\exists \param} \varphi$ when
the parameter $\param$ ranges over $\RP$ but not when $\param$ ranges
over $\mathbb{N}$.
\end{example}

Let $\A=\tuple{\locs,\clocks,\edges}$ be a timed automaton augmented
with a labelling function~$\lbl:\locs\rightarrow 2^{\prop}$.  Let
$\varphi$ be a {\ptctl} formula in which all occurrences of parameters
are bound.  The model checking problem of~$\A$ against~$\varphi$ asks,
given a configuration~$\tuple{\loc,\val}$ of~$\A$, whether
$\tuple{\loc,\val} \models \varphi$.

The model checking procedure for parametric {\tctl} with
integer-valued parameters, developed in~\cite{BruyereDR08}, relies on
the region abstraction.  In particular, formulas in this logic have
the same truth value for all configurations in a given region.
However, as the following example shows, region invariance fails when
parameters range over the set of real numbers.
\begin{example}
Consider the timed automaton in Figure~\ref{fig:realParameters}.  Let
$\varphi = \exists \param\, \exists\F_{=\param} (p_1 \wedge
\exists\F_{=\param} p_2)$.  Then a configuration
$(\loc_0,\vect{t_1\\t_2})$ satisfies $\varphi$ just in case $t_1,t_2<1$ and
$2t_1-t_2<1$, for $\param=(1-t_2)/2$.
\end{example}
In Section~\ref{sec:MC} we show that model checking {\ptctl} over
real-valued parameters is decidable in EXPSPACE and it is NEXPTIME-hard.

\section{Difference Bound Matrices}  
\subsection{Basic Definitions}\label{sec:DBM}
In this section we review the notions of clock zones and difference
bound matrices; see \cite{Clarke99,Bouyer03} for further details.  

Let $\clocks = \{x_1,\ldots,x_n\}$ be a set of clock variables.
A \emph{zone} $\zone \subseteq \RP^{\clocks}$ is a set of 
valuations defined by a conjunction of \emph{difference constraints}
$x_j-x_i \prec c$ for $c \in \mathbb{R}$ and ${\prec} \in \{<,\leq\}$.
Note that we allow real-valued constants in difference constraints.

Zones and operations thereon can be efficiently represented
using \emph{difference bound matrices} (DBMs).
A DBM is an $(n+1) \times (n+1)$ matrix~$M$ with entries in the set
\[\mathbb{V}=(\{<,\leq\} \times \mathbb{R}) \cup \{(<,\infty)\} \, .\]
A DBM $M=(\prec_{i,j},m_{i,j})$ can be interpreted as a
conjunction of constraints $x_i-x_j\prec_{i,j} m_{i,j}$, where $x_0$
is a special clock that symbolically represents zero.  Formally, the
semantics of DBM~$M$ is the zone
$$
\semantics{M}=\Big\{\val\in \RP^{\clocks} : 
\bigwedge_{0\leq i,j\leq n}\val_i-\val_j\prec_{i,j} m_{i,j}\Big\} \, ,
$$ 
where $\nu_0=0$.  Figure~\ref{fig:d0} depicts a zone~$\zone \subseteq
[0,1]^2$ containing a line segment and a DBM~$M$ with
$\semantics{M}=\zone$ .

An \emph{atomic} DBM~$M'$ is one that represents a single
constraint~$x_i-x_j\sim c$, where ${\sim} \in \{<,\leq\}$ and $c\in
\mathbb{R}$.  Note that all but one entry of an atomic DBM is the
trivial constraint~$(<,\infty)$.  We often denote DBMs by the
constraints that they represent.

Define a total order~$\leq_{\mathbb{V}}$ on~$\mathbb{V}$ by writing
$(\prec,m)\leq_{\mathbb{V}} (\prec',m')$ if $m<m'$ or if $m=m'$ and
either ${\prec}=<$ or ${\prec'}={\leq}$.  Define addition
on~$\mathbb{V}$ by $ (\prec,m)+ (\prec',m')=(\prec'',m+m')$, where
$$ {\prec''} =
\begin{cases}
\leq & \text {if } {\prec}={\leq}  \text{ and } {\prec'}={\leq}, \\   
< & \text{otherwise}.\\
  \end{cases}
$$ Here we adopt the convention that $m+\infty=\infty+m=\infty$ for
all $m\in\mathbb{R}$.  A DBM~$M=(M_{i,j})$ is in \emph{canonical form}
if~$M_{i,k}\leq_{\mathbb{V}} M_{i,j}+M_{j,k}$ for all~$0\leq i,j,k
\leq n$.  One can transform an arbitrary DBM into an equivalent
canonical-form DBM using the Floyd-Warshall algorithm.  For all
non-empty clock zones~$\zone$, there is a unique DBM~$M$ in canonical
form with~$\semantics{M}=\zone$.  A DBM~$M$ is said to be \emph{consistent} if
$\semantics{M}\neq \emptyset$.  If $M$ is in canonical form, then it is
consistent if and only if $(\leq,0) \leq_{\mathbb{V}} M_{i,i}$ for all~$0\leq i\leq n$.

We now define operations on DBMs that correspond to
 time elapse, projection, and intersection on zones.

\textbf{Time Elapse.}
The image of a DBM $M$ under \emph{time elapse} 
is the DBM $\overrightarrow{M}$ defined by
\[ \overrightarrow{M}_{i,j} = \left\{\begin{array}{ll}(<,\infty) & 
	\text{if $i\neq0, j=0$}\\[2pt]
M_{i,j}  & \text{otherwise.}
\end{array} \right. \]
If $M$ is canonical, then $\overrightarrow{M}$ is also canonical and we have
$\semantics{ \overrightarrow{M}} =\{ \nu + t : \nu \in 
\semantics{M} \mbox{ and }t\geq 0\}$.\\

\textbf{Reset.}
		The image of a DBM $M$ under \emph{resetting clock $x_\ell$}
is~$M[x_\ell \leftarrow 0]$, given by
$M[x_\ell \leftarrow 0]_{i,j} = M_{i_\ell,j_\ell},$ where for any index $k$,
\[ k_\ell = \left\{ \begin{array}{ll} k & \mbox{ if $k\neq\ell$}\\[2pt]
                            0 & \mbox{ otherwise.}
\end{array}\right.  \]
If $M$ is canonical, then $M[x_\ell \leftarrow 0]$ is also canonical
and $\semantics{M} = \{ \nu[x_\ell \leftarrow 0] : \nu \in
\semantics{M} \}$.\\

\textbf{Intersection.}
Our presentation of intersection of DBMs is slightly non-standard.
First, we only consider intersection with atomic DBMs.  (Clearly this
is without loss of generality since any DBM can be written as an
intersection of atomic DBMs.)  Under this restriction we combine
intersection and canonisation, so that our intersection operation
yields a DBM in canonical form if the input DBM is in canonical form.
Specifically, let $M'$ be an atomic DBM with non-trivial
constraint~$M'_{p,q}$.  The DBM $M''= M \cap M'$ is given by
\[ M''_{i,j} = \min(M_{i,j},M_{i,p}+M'_{p,q}+M_{q,j}) \]
for all~$i,j$.   	
Then $M''$ is canonical and 
$\semantics{M''} = \semantics{M} \cap \semantics{M'}$.

\subsection{Closure of a DBM}
\label{sec:closure}
We will use zones to represent the fractional parts of clocks in a
given set of valuations.  For this reason we are solely interested in
zones contained in $[0,1]^n$.  We say that a DBM~$M$
is \emph{$1$-bounded} if for all entries $(\prec,m)$ of $M$ we have
$-1 \leq m \leq 1$.  It is clear that if $M$ is $1$-bounded then
$\semantics{M} \subseteq [0,1]^n$.  Conversely the unique DBM in
canonical form that represents a zone $Z\subseteq [0,1]^n$ is
necessarily $1$-bounded since the constraints in a canonical DBM
cannot be tightened.

Given a $1$-bounded DBM $M$, define the \emph{closure} $M$ to
be the smallest set $\closure(M)$ of DBMs such that
$M \in \closure(M)$, and if $N\in \closure(M)$ then
\begin{itemize}
\item $N \cap M' \in \closure(M)$ for all atomic DBMs $M'$
with numerical entries in $\mathbb{Z}\cup\{\infty\}$.
\item 
$\overrightarrow{N} \cap \bigcap_{i=1}^n (x_i \leq 1) \in \closure(M)$,
\item 
$N[x_i\leftarrow 0] \in \closure(M)$ for $0\leq i \leq n-1$,
\item 
$(N \cap (x_n=1))[x_n \leftarrow 0] \in \closure(M)$.
\end{itemize}
We make three observations about this definition.  First, notice that
in the first item we only require closure with respect to intersection
with constraints with integer constants.  Observe also that in the
second item the time elapse operation has been relativized to
$[0,1]^n$.  This ensures that every DBM $N \in \closure(M)$ denotes a
subset of $[0,1]^n$.  It follows that any consistent DBM in
$\closure(M)$ is 1-bounded.  Finally, note that the clock $x_n$ is
treated in a special way (in keeping with our assumptions about timed
automata in Section~\ref{sec:TA}): it is only reset when it reaches
$1$.

Let $\val \in [0,1]^n$ be a clock valuation, and recall that, by
 convention, $\val_0=0$.  We write $M_\val$ for the 1-bounded DBM $M_\val =
 (\prec_{i,j},m_{i,j})$, where ${\prec_{i,j}}={\leq}$ and
 $m_{i,j}=\val_j-\val_i$ for all $0\leq i,j\leq n$.  Then $M_\val$ is in
 canonical form and $\semantics{M_\val}=\{\val\}$.  

We say a DBM~$M=(\prec_{i,j},m_{i,j}) \in\closure(M_\val)$
is \emph{well-supported}, if each entry $m_{i,j}$ can be written in
the form $c+\val_{j'}-\val_{i'}$ for some $c\in\{-1,0,1\}$ and
indices $0\leq i',j'\leq n$.  Clearly $M_{\val}$ is well-supported.

The following is the main technical result in this section.  See
Appendix~\ref{append-lem-support} for the full proof.
\begin{lemma}\label{lem:support}
  Let $\val \in [0,1]^n$ be a clock valuation.  Then every consistent DBM
  lying in $\closure(M_\val)$ is well-supported. 
\end{lemma}

\begin{proof}[Proof Sketch]
  We show by induction on the structure of~$\closure(M_\val)$ that any
  consistent DBM $M \in\closure(M_\val)$ is well-supported.  The key
  case is for intersection (see Section~\ref{sec:DBM}), which
  does not immediately preserve well-supportedness due to the possibility
  that ~$M''_{i,j}=M_{i,p}+M'_{p,q}+M_{q,j}$.  However we show that in
  this case at least one of $m_{i,p}$ or $m_{q,j}$ lies in
  $\mathbb{Z}$, which ensures well-supportedness of $M''$.
\end{proof}

\subsection{Parametric DBMs}
\label{symbolic}
In this subsection we observe that the construction of
$\closure(M_{\nu})$ can be carried out parametrically, based on the
logical \emph{type} of the clock valuation $\nu\in[0,1]^n$ (to be
defined below).  In particular, if $\nu,\nu' \in [0,1]^n$ have the
same type, then $\closure(M_{\nu})$ and $\closure(M_{\nu'})$ can both
be seen as instances of a common parametric construction.

Recall from Subsection~\ref{sec:Logic} the definition of the set of
terms~$\mathcal{T}_{\mathbb{R}}$ of real arithmetic.  Given
$n\in\mathbb{N}$, let us further write $\mathcal{T}_{\mathbb{R}}(n)$
for the set of terms in variables $r_0,\ldots,r_n$.  A valuation
$\nu\in [0,1]^n$ extends in a natural way to a function
$\nu:\mathcal{T}_{\mathbb{R}}(n)\rightarrow \mathbb{R}$ mapping $r_i$
to $\nu_i$ (recalling the convention that $\nu_0=0$).

Given a clock valuation $\nu \in [0,1]^n$, the \emph{type} of $\nu$ is
the set of atomic $\mathcal{L}_{\mathbb{R}}$-formulas $t\leq t'$, with
$t,t'\in \mathcal{T}_{\mathbb{R}}(n)$ that are satisfied by the
valuation $\nu$. A collection of atomic formulas~$\tau$ is said to be
an \emph{$n$-type} if it is the type of some clock valuation~$\nu\in
[0,1]^n$.  Note that every type contains the inequalities~$r_0\leq 0$
and~$0\leq r_0$.

Given an $n$-type $\tau$, we define an equivalence relation on the set
of terms $\mathcal{T}_{\mathbb{R}}(n)$ that relates terms $t$ and $t'$
just in case the formulas $t\leq t'$ and $t'\leq t$ both lie in
$\tau$.  We write $[t]$ for the equivalence class of term $t$ and
denote by $\mathcal{T}_{\mathbb{R}}(\tau)$ the set of equivalence
classes of $\mathcal{T}_{\mathbb{R}}(n)$.  We can define a linear
order on $\mathcal{T}_{\mathbb{R}}(\tau)$ by writing $[t]\leq [t']$ if
and only if formula $t\leq t'$ lies in $\tau$.  We define an addition
operation on $\mathcal{T}_{\mathbb{R}}(\tau)$ by writing
$[t]+[t']=[t+t']$.

Given an $n$-type $\tau$, a \emph{parametric DBM} of dimension $n$
over~$\mathcal{T}_{\mathbb{R}}(\tau)$ is an $(n+1)\times(n+1)$ matrix
with entries in
\[(\{ <,\leq \} \times \mathcal{T}_{\mathbb{R}}(\tau)) \cup \{
(<,\infty)\} \, .\] We use letters in calligraphic font to denote
parametric DBMs, and roman font for concrete DBMs.  Given a parametric
DBM $\mathcal{M}$, we obtain a concrete DBM $\nu(\mathcal{M})$ by
applying $\nu$ pointwise to the entries of $\mathcal{M}$.

The time elapse and reset operations on DBMs, defined in
Section~\ref{sec:DBM}, formally carry over to parametric DBMs.  Since
the notions of addition and minimum are well-defined on
$\mathcal{T}_{\mathbb{R}}(\tau)$, we can also formally carry over the
definition of intersection to parametric DBMs.

\begin{proposition}
Let $\nu \in [0,1]^n$ be a clock valuation with type~$\tau$ and let
$\mathcal{M}$ be a parametric DBM over $\mathcal{T}_{\mathbb{R}}(\tau)$.
 Then
\begin{enumerate}
\item $\nu(\overrightarrow{\mathcal{M}}) =
\overrightarrow{\nu(\mathcal{M})}$.
\item $\nu(\mathcal{M}[x_i \leftarrow 0]) =
\nu(\mathcal{M})[x_i \leftarrow 0]$.
\item $\nu(\mathcal{M} \cap N) = \nu(\mathcal{M})\cap N$ for all
atomic DBMs $N$.
\end{enumerate}
\label{prop:homo}
\end{proposition}
\begin{proof}
  Suppose that $\val$ has type $\tau$.  Then
  $\val:\mathcal{T}_{\mathbb{R}}(\tau) \rightarrow \mathbb{R}$ is an
  order embedding ($[t] \leq [t']$ if and only if $\val(t) \leq
  \val(t')$) and a homomorphism ($\val([t]+[t']) = \val([t])+\val([t])$).
  In particular, $\val$ preserves all operations used to define time
  elapse, projection, and intersection of DBMs.  The result follows.
\end{proof}

Since the basic operations on DBMs are all defined for parametric
DBMs, we can also formally carry over the definition of the closure of
a DBM to parametric DBMs.  In particular, given an $n$-type $\tau$, we
consider the closure of the parametric DBM
$\mathcal{M}_{\tau}=(\prec_{i,j},m_{i,j})$
over~$\mathcal{T}_{\mathbb{R}}(\tau)$, where ${\prec_{i,j}}={\leq}$
and $m_{i,j}=[r_i-r_j]$.  Note that $\val(\mathcal{M}_{\tau}) =
M_{\val}$ for any clock valuation $\val \in [0,1]^n$.  Then, by
Proposition~\ref{prop:homo}, we have the following result:
\begin{proposition}\label{prop:closure-tau}
Let $\val \in [0,1]^n$ be a clock valuation with type~$\tau$.  Then
\[ \left\{ \val(\mathcal{M}) : \mathcal{M} \in 
\closure(\mathcal{M}_{\tau}) \right\} 
= \closure(M_{\val}) \, .\]
\end{proposition}

Define the set $\mathcal{DT}_{\mathbb{R}}(n)$ of \emph{difference
  terms} to be the subset of $\mathcal{T}_{\mathbb{R}}(n)$ comprising
those terms of the form $c+r_i-r_j$, where $c\in\{-1,0,1\}$ is a
constant and $r_i,r_j$ are variables with $0 \leq i,j \leq n$.  From
Lemma~\ref{lem:support} and Proposition~\ref{prop:closure-tau} we now
have:

\begin{corollary}
  Fix an $n$-type $\tau$.  Then every DBM in
  $\closure(\mathcal{M}_{\tau})$ has all its entries of the form
  $(\prec,[t])$, where ${\prec}\in\{<,\leq\}$ and $t \in
  \mathcal{DT}_{\mathbb{R}}(n)$.
\label{corl:well-supported}
\end{corollary}

The significance of Corollary~\ref{corl:well-supported} is that the
only part of the type $\tau$ required to determine
$\closure(\mathcal{M}_{\tau})$ is the \emph{finite} collection of
formulas $t\leq t'$ in $\tau$ such that
$t,t' \in \mathcal{DT}_{\mathbb{R}}(n)$.  Thus
$\closure(\mathcal{M}_\tau)$ is finite.  Indeed it is not hard to see
from Corollary~\ref{corl:well-supported} that
$|\closure(\mathcal{M}_{\tau})| \leq 2^{\mathrm{poly}(n)}$.

\section{A Family of Region Automata}
Let $\A$ be a timed automaton.  Our aim in this section is to define a
finite collection of counter machines that represents the reachability
relation on $\A$.  Intuitively the counters in these machines are used
to store the integer parts of clock valuations of reachable
configurations, while the fractional parts of the clock valuations are
aggregated into zones which are represented by difference bound
matrices encoded within control states.

\subsection{Monotonic Counter Machine}
\label{sec:monotone}
In this subsection we introduce the class of \emph{monotonic counter
  machines} and show that the reachability relation for a machine
in this class is definable in Presburger arithmetic.  The proof is
straightforward, and is related to the fact that the reachability
relation of every reversal-bounded counter machine is Presburger
definable~\cite{FinkelS08}.

Let $C = \{c_1,\ldots,c_n\}$ be a finite set of \emph{counters}.
The collection of \emph{guards}, denoted $\Phi(C)$, is given by the grammar
\[ \varphi ::= \true \mid 
c < k \mid c=k \mid c>k \mid \varphi \wedge \varphi \, , \] 
where $c\in C$ and $k\in\mathbb{Z}$.  The set of
\emph{counter operations} is
\[ \mathrm{Op}(C)=\{ \mathit{reset}(c) , \mathit{inc}(c) : c \in C \} \cup \{ \mathit{nop}\} \, . \] 
A \emph{monotonic counter machine} is a tuple
$\mathcal{C}=\langle S,C,\Delta\rangle$, where $S$ is a
finite set of \emph{states}, $C$ is a finite set of
\emph{counters}, and
$\Delta \subseteq S\times \Phi(C) \times
\mathrm{Op}(C) \times S$ is a set of \emph{edges}.

The set of \emph{configurations} of $\mathcal{C}$ is
$S\times\mathbb{N}^n$.  A configuration $\tuple{s,\upsilon}$ consists
of a state~$s\in S$ and a \emph{counter valuation}~$\upsilon \in
\mathbb{N}^n$, where $\upsilon_i$ represents the value of
counter~$c_i$ for $i=1,\ldots,n$. 
The satisfaction relation~$\models$ between counter valuations and guards
is defined in the obvious way.
The \emph{transition relation}  
\[{\rightarrow} \subseteq (S\times\mathbb{N}^n)\times
  (S\times\mathbb{N}^n) \] is specified by writing
$\tuple{s,\upsilon}\rightarrow \tuple{s',\upsilon'}$ just in case
at least one of the following holds:
\begin{itemize}
\item there is an edge
  $\tuple{ s,\varphi,\mathit{nop},s'} \in \Delta$ such that
  $\upsilon\models \varphi$ and $\upsilon=\upsilon'$;
\item there is an edge
$\tuple{ s,\varphi,\mathit{reset}(c_i),s'} \in \Delta$
such that 
$\upsilon \models \varphi$, $\upsilon'_i=0$, and
$\upsilon'_j=\upsilon_j$ for $i\neq j$;
\item there is an edge
  $\tuple{ s,\varphi,\mathit{inc}(c_i),s'} \in \Delta$ such that
  $\upsilon \models \varphi$, $\upsilon'_i=\upsilon_i+1$, and $\upsilon'_j=\upsilon_j$ for $i\neq j$.
\end{itemize}
The reachability relation on $\mathcal{C}$ is the reflexive transitive
closure of $\rightarrow$.  

The proof of the following result is given in Appendix~\ref{append:monotone}.
\begin{proposition} \label{prop:monotone}
  Let $\mathcal{C}$ be a monotonic counter machine with~$n$ counters.
  Given states $s,s'$ of $\mathcal{C}$, the reachability relation \[
  \{ \tuple{\upsilon,\upsilon'} \in \mathbb{N}^{2n} :
  \tuple{s,\upsilon} \mathrel{\longrightarrow^*} \tuple{s',\upsilon'}
  \} \] is definable by a formula in the existential fragment of
  Presburger arithmetic that has size exponential in $\mathcal{C}$.
\end{proposition}

\subsection{Concrete Region Automata}
\begin{figure*}[t!]
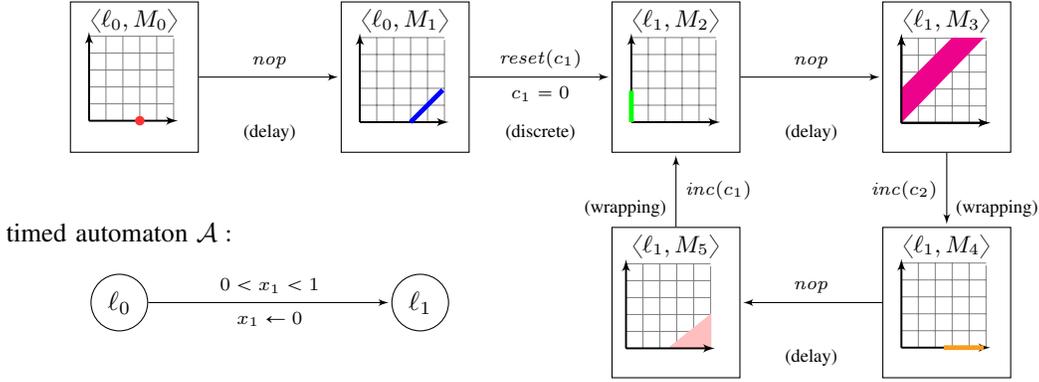


\centering

\begin{tikzpicture}[>=latex',shorten >=1pt,node distance=1.9cm,on grid,auto,
state/.style={
           rectangle,
           draw=black, 
           minimum height=2cm,
					minimum width=1.7cm,
           inner sep=2pt,
           text centered,
           },
roundnode/.style={circle, draw,minimum size=1.2mm},]

\node [state,label={[label distance=.3cm]90:counter machine $\mathcal{C}_{\tuple{\loc_0,\val}}$:}] (l0z0) at(0,0) {~\scalebox{.9}{\input{z0.tex}}};
\node [draw=none] (dum) [above=.75 of l0z0] {{\small $\tuple{\ell_0,M_0}$}};

\node [state] (l0z1) [right=3.6cm of l0z0] {~\scalebox{.9}{\input{z1.tex}}};
\node [draw=none] (dum) [above=.75 of l0z1] {{\small$ \tuple{\ell_0, M_1}$}};

\node [state] (l1z2) [right=3.6cm of l0z1] {~\scalebox{.9}{\input{z2.tex}}};
\node [draw=none] (dum) [above=.75 of l1z2] {{\small $\tuple{\ell_1,M_2}$}};

\node [state] (l1z3) [right=3.6cm of l1z2] {~\scalebox{.9}{ \begin{tikzpicture}
    \draw[very thin,color=gray,step=.25cm] (0,0) grid (1.25,1.25);    
    \draw[->,thick] (0,0) -- (1.35,0);
    \draw[->,thick] (0,0) -- (0,1.35);

\draw [draw=magenta, fill=magenta]
     (0,0)--  (0,.5) -- (.75,1.25) -- (1.25,1.25);
		    \draw[->,thick] (0,0) -- (1.35,0);
    \draw[->,thick] (0,0) -- (0,1.35);

\end{tikzpicture}}};
\node [draw=none] (dum) [above=.75 of l1z3] {{\small $\tuple{\ell_1,M_3}$}};

\node [state] (l1z4) [below=3cm of l1z3] {~\scalebox{.9}{\input{z4.tex}}};
\node [draw=none] (dum) [above=.75 of l1z4] {{\small $\tuple{\ell_1,M_4}$}};

\node [state] (l1z5) [left=3.6cm of l1z4] {\scalebox{.9}{\input{z5.tex}}};
\node [draw=none] (dum) [above=.75 of l1z5] {{\small $\tuple{\ell_1,M_5}$}};

\path[->] (l0z0) edge node [above,midway] {\scriptsize{$\mathit{nop}$}}  node[above,below=.5cm] {\scriptsize{(delay)}} (l0z1);
\path[->] (l0z1) edge node [above,midway] {\scriptsize{$\mathit{reset}(c_1)$}} node [below,midway] {\scriptsize{$c_1=0$}} node[above,below=.5cm] {\scriptsize{(discrete)}} (l1z2);
\path[->] (l1z2) edge node [above,midway] {\scriptsize{$\mathit{nop}$}} node[above,below=.5cm] {\scriptsize{(delay)}} (l1z3);
\path[->] (l1z3) edge node [left,midway] {\scriptsize{$\mathit{inc}(c_2)$}} node[near end,right] {\scriptsize{(wrapping)}} (l1z4);
\path[->] (l1z4) edge node [above,midway] {\scriptsize{$\mathit{nop}$}} node[above,below=.5cm] {\scriptsize{(delay)}} (l1z5);
\path[->] (l1z5) edge node [right,midway] {\scriptsize{$\mathit{inc}(c_1)$}}  node[near start,left] {\scriptsize{(wrapping)}} (l1z2);

\node [roundnode,label={[label distance=.3cm]90:timed automaton $\A:$}] (l0) at(-.2,-3) {$\ell_0$};
\node [roundnode](l1)  [right=4cm of l0] {$\ell_1$};

\path[->] (l0) edge  node [above,midway] {\scriptsize{$0<x_1< 1$}}  node [below,midway] {\scriptsize{$x_1\leftarrow 0$}} (l1);

\end{tikzpicture}

\caption{A timed automaton~$\A$ together with  
  the fragment of counter machine~$\mathcal{C}_{\tuple{\loc_0,\val}}$ 
	relevant to 
	expressing the reachability relation of~$\ell_0$ and~$\ell_1$.  The
  valuation~$\val$ is such that $\val_1=0.6$ and $\val_2=0$.
  States~$\tuple{\ell,M}$ of the counter machine are illustrated by
  $\ell$ and the zone that~$M$ represents.  The initial state is
  $\tuple{\loc_0,M_0}$, where $M_0=M_{\val}$.
}
\label{fig:reachAutomata}
\end{figure*}

Let $\A=\tuple{\locs,\clocks,\edges}$ be a timed automaton and
$\tuple{\loc,\val}$ a configuration of $\A$.
We define a monotonic counter machine~$\mathcal{C}_{\tuple{\loc,\val}}$ 
whose configuration graph represents
all configurations of $\A$ that are reachable from~$\tuple{\loc,\val}$.

Let $\clocks=\{x_1,\ldots,x_n\}$ be the set of clocks in~$\A$.  Recall
from Section~\ref{sec:TA} the assumption that clock $x_n$ is never
reset by the timed automaton.  To simplify the construction, we also
assume that each transition in~$\A$ resets at most one clock.  This is
without loss of generality with respect to reachability.

Given a clock constraint $\varphi \in \Phi(\clocks)$, we decompose
$\varphi$ into an integer constraint $\varphi_{\mathsf{int}} \in
\Phi(C)$ and a real constraint $\varphi_{\mathsf{frac}} \in
\Phi(\clocks)$ such that for every clock valuation $\val' \in
\RP^{\clocks}$,
\[ \val' \models \varphi \quad\mbox{iff}\quad \lfloor \val' \rfloor \models \varphi_{\mathsf{inc}}
	\mbox{ and } \fract(\val') \models \varphi_{\mathsf{frac}} \,\] 
The definition of $\varphi_{\mathsf{int}}$ and
        $\varphi_{\mathsf{frac}}$ is by induction on the structure of
        $\varphi$.  The details are given in
        Figure~\ref{fig:decompose}.  

\begin{figure}[ht]
\begin{center}
\begin{tabular}{ c| c c c c} 
$\varphi$ & $x < k$ & $x= k$ & $k<x<k+1$& $x\geq k $  \\
\hline
$\varphi_{{\sf int}}$ & $c\leq k-1$ & $c=k$ & $c=k$ & $c\geq k$  \\ 
$\varphi_{{\sf frac}}$ &$x<1$ & $x=0$& $0<x<1$& $x\geq 0$  \\  
\end{tabular}
\end{center}
\caption{Base cases of the inductive definition of
  $\varphi_{\mathsf{inc}}$ and $\varphi_{\mathsf{frac}}$, where $x$ is
  a clock variable and $c$ is a counter variable.  (Note any guard
  $\varphi\in \Phi(X)$ can be expressed as a Boolean combination of
  the basic guards in the table.)  For the inductive step we have
  $(\varphi \wedge \varphi')_{\sf int} = \varphi_{\sf int}\wedge
  \varphi'_{\sf int}$ and $(\varphi \wedge \varphi')_{\sf frac} =
  \varphi_{\sf frac}\wedge \varphi'_{\sf frac}$.}
\label{fig:decompose}
\end{figure}

The construction of the counter machine
$\mathcal{C}_{\tuple{\loc,\val}}=\tuple{S,C,\Delta}$ is such that the
set~$S$ of states comprises all pairs~$\tuple{\loc', M}$ such that
$\loc'\in \locs$ is a location of~$\A$ and $M\in\closure(M_{\fract(\val)})$ is a consistent DBM.  The set of counters
is $C=\{c_1,\ldots,c_n\}$, where $n$ is the number of clocks in~$\A$.
Intuitively the purpose of counter $c_i$ is to store the integer part
of clock $x_i$, for $i=1,\ldots,n$.

We classify the transitions of
$\mathcal{C}_{\tuple{\loc, \val}}$ into three different types:
 From all states $\tuple{\ell_1,M_1}$ to a state $\tuple{\ell_1,M_2}$, there is  
\begin{itemize}
\item a \emph{delay transition} if
  $M_2=\overrightarrow{M_1} \cap \bigcap_{i=1}^n (x_i \leq 1)$.  Such
  a transition has guard $\true$ and operation $\mathit{nop}$;
\item  a \emph{wrapping transition} if 
$M_2 = (M_1\cap (x_i=1))[x_i \leftarrow 0]$ for some clock~$x_i$.  Such a transition
has guard $\true$ and operation $\mathit{inc}(c_i)$.
\end{itemize}
Suppose that $(\loc,\varphi,\{x_i\} ,\loc')$ is a transition of $\A$.
Decompose the guard $\varphi$ into~$\varphi_{{\sf int}}$ and
$\varphi_{{\sf frac}}$. Then from all states $\tuple{\ell_1,M_1}$ to a
state~$\tuple{\ell_2,M_2}$, there is
\begin{itemize}
\item  a \emph{discrete transition} if 
$M_2=(M_1\cap \varphi_{{\sf frac}})[x_i \leftarrow 0]$. Such a transition
has guard $\varphi_{{\sf int}}$ and operation $\mathit{reset}(c_i)$.
\end{itemize}

The following proposition describes how the set of reachable
configurations in $\mathcal{C}_{\tuple{\loc,\val}}$ represents the set
of configurations reachable from $\tuple{\loc,\val}$ in the timed
automaton $\A$.  The proposition is a straightforward variant of the
soundness and completeness of the DBM-based forward reachability
algorithm for timed automata, as shown, e.g., in~\cite[Theorem
  1]{BengtssonY03}.  We give a proof in Appendix~\ref{app:sound}.

\begin{proposition}
  Configuration~$\tuple{\loc',\val'}$ is reachable
  from~$\tuple{\loc,\val}$ in~$\A$ if and only if there exists some
  DBM $M' \in \closure(M_{\mathrm{frac}(\val)})$ such that the
  configuration~$\tuple{\tuple{\loc',M'}, \floor{\val'}}$ is reachable
  from~$\tuple{\tuple{\loc,M_{\fract(\val)}},\floor{\val}}$ in the
  counter machine~$\mathcal{C}_{\tuple{\loc,\val}}$ and
  $\mathrm{frac}(\val')\in \semantics{M'}$.
\label{prop:sound}
\end{proposition}

We illustrate the translation from timed automata to counter machines
with the following example.
\begin{example}
Consider the timed automaton~$\A$ in Figure~\ref{fig:reachAutomata}
with clocks~$\clocks=\{x_1,x_2\}$, where $x_2$ is the reference
clock. Let the configuration~$\tuple{\loc_0,\val}$ be such that
$\val=\vect{0.6\\0}$.  Also shown in Figure~\ref{fig:reachAutomata} is
the counter machine $\mathcal{C}_{\tuple{\loc_0,\val}}$ that is
constructed from $\A$ and $\tuple{\loc_0,\val}$ in the manner
described above.  The control states of this machine are pairs
$\tuple{\loc,M}$, where $\loc$ is a location of $\A$ and $M$ is a
consistent DBM in $\closure(M_{\val})$. The
machine~$\mathcal{C}_{\tuple{\loc_0,\val}}$ has two counters, respectively
denoted by $c_1$ and $c_2$.

The initial state of $\mathcal{C}_{\tuple{\loc_0,\val}}$ is
$\tuple{\loc_0,M_0}$, where $M_0 = M_{\val}$.  Note that
  $\semantics{M_0} = \left\{ \vect{0.6\\0}\right\}$.  The
  counter-machine state~$\tuple{\loc_0,M_0}$ in tandem with counter
  valuation $\vect{0\\0}$ represents the configuration $\tuple{\ell_0,\val}$
  of $\A$.

There is a delay edge in $\mathcal{C}_{\tuple{\loc_0,\val}}$ from
$\tuple{\loc_0,M_0}$ to $\tuple{\loc_0,M_1}$, where $M_1
=\overrightarrow{M_0} \cap \bigcap_{i=1}^2 (x_i \leq 1)$.  We then have
$\semantics{M_1} = \left\{ \vect{0.6\\0} + t : 0 \leq t \leq 0.4
\right\}$.  

The single transition of $\A$ yields a discrete edge in
$\mathcal{C}_{\tuple{\loc_0,\val}}$ from $\tuple{\loc_0,M_1}$ to
$\tuple{\loc_1,M_2}$.  This transition in $\A$ has guard $\varphi \defequals
0<x_1<1$.  This decomposes into separate constraints on the integer
and fractional parts, respectively given by
\[\begin{aligned} 
\varphi_{{\sf int}}\defequals (c_1=0) \qquad \text{ and } \qquad
\varphi_{{\sf frac}}\defequals(0<x_1<1).
\end{aligned}
\] 
The integer part $\varphi_{\mathsf{int}}$ becomes the guard of the
corresponding edge in $\mathcal{C}_{\tuple{\loc_0,\val}}$.  The
fractional part $\varphi_{\mathsf{frac}}$ is incorporated into the DBM
$M_2$, which is defined as 
$$M_2 = (M_1 \cap (0<x_1<1))[x_1 \leftarrow
  0],$$  
where $\semantics{M_2} = \left\{ \vect{0\\y} : 0 \leq y
< 0.4 \right\}$.  
There is a further delay edge in $\mathcal{C}_{\tuple{\loc_0,\val}}$ 
from $\tuple{\loc_1,M_2}$ to $\tuple{\loc_1,M_3}$.

There is a wrapping edge from $\tuple{\loc_1,M_3}$ to
$\tuple{\loc_1,M_4}$, where $M_4 = (M_3 \cap (x_2=1))[x_2 \leftarrow
  0]$.  The counter $c_2$ is incremented along this edge,
corresponding to the integer part of clock $x_2$ increasing by $1$ as
time progresses.

The remaining states and edges of $\mathcal{C}_{\tuple{\loc_0,\val}}$
are illustrated in Figure~\ref{fig:reachAutomata}.  Note that we only
represent states 
that are relevant to  expressing reachability from $\ell_0$ to~$\ell_1$. 
\end{example}

An important fact about the collection of counter
machines~$\mathcal{C_{\tuple{\loc,\val}}}$, as $\fract(\val)$ varies over
$[0,1]^{\clocks}$, is that there are only finitely many such machines
up to isomorphism.  This essentially follows from
Proposition~\ref{prop:closure-tau}, which shows that
$\closure(M_{\fract(\val)})$ is determined 
by the type of $\fract(\val)$. In the next
section we develop this intuition to build a symbolic counter machine
that embodies $\mathcal{C_{\tuple{\loc,\val}}}$ for all valuations
$\val$ of the same type.

\subsection{Parametric Region Automata }
Consider a timed automaton $\A$ with $n$ clocks, a location $\loc$ of
$\A$, and an $n$-type $\tau$.  In this section we define a monotonic
counter machine~$\mathcal{C}_{\tuple{\loc,\tau}}$ that can be seen
as a parametric version of the counter machine
$\mathcal{C}_{\tuple{\ell,\val}}$ from the previous section, where
valuation $\val$ has type $\tau$.

First recall that $\mathcal{M}_\tau = (\prec_{i,j},m_{i,j})$ is the
parametric DBM over $\mathcal{T}_{\mathbb{R}}(\tau)$ such that
$\prec_{i,j}=\leq$ and $m_{i,j}=[r_i-r_j]$ for $0 \leq i,j \leq n$.

The construction of the counter machine
$\mathcal{C}_{\tuple{\loc,\tau}}$ is formally very similar to that of
$\mathcal{C}_{\tuple{\loc,\val}}$. Specifically, the set~$S$ of states
of $\mathcal{C}_{\tuple{\loc,\tau}}$ comprises all
pairs~$\tuple{\loc', \mathcal{M}'}$ such that~$\loc'\in \locs$ is a
location in~$\A$ and $ \mathcal{M}' \in \closure(\M_{\tau})$ is a
consistent parametric DBM.  The set of counters is
$C=\{c_1,\ldots,c_n\}$, where $n$ is the number of clocks in~$\A$.
The transitions of~$\mathcal{C}_{\tuple{\loc,\tau}}$ are defined in a
formally identical way to those of~$\mathcal{C}_{\tuple{\loc,\val}}$;
we simply replace operations on concrete DBMs with the corresponding
operations on parametric DBMs.

With the above definition, it follows from Proposition~\ref{prop:homo}
that the counter machine $\mathcal{C}_{\tuple{\loc,\tau}}$ and
$\mathcal{C}_{\tuple{\loc,\val}}$ are isomorphic via the map sending a
control state $\tuple{\loc,\mathcal{M}}$ of $\mathcal{C}_{\tuple{\loc,\tau}}$ to
the control state $\tuple{\loc,\val(\mathcal{M})}$ of
$\mathcal{C}_{\tuple{\loc,\val}}$.  Proposition~\ref{prop:sound} then
yields:

\begin{theorem}
Consider states~$\tuple{\loc,\val}$ and $\tuple{\loc',\val'}$ of a
timed automaton~$\A$ such that $\fract(\val)$ has type~$\tau$.  Then
$\tuple{\loc',\val'}$ is reachable from $\tuple{\loc,\val}$ in $\A$ if
and only if there exists some DBM $\M' \in \closure(\M_\tau)$ such
that the configuration~$\tuple{\tuple{\loc',\M'}, \floor{\val'}}$ is
reachable from~$\tuple{\tuple{\loc,\M_{\tau}},\floor{\val}}$ in the
counter machine~$\mathcal{C}_{\tuple{\loc,\tau}}$ and
$\mathrm{frac}(\val') \in \semantics{\mathrm{frac}(\val)(\M')}$.
\label{thm:symbolic}
\end{theorem}

\subsection{Reachability Formula}

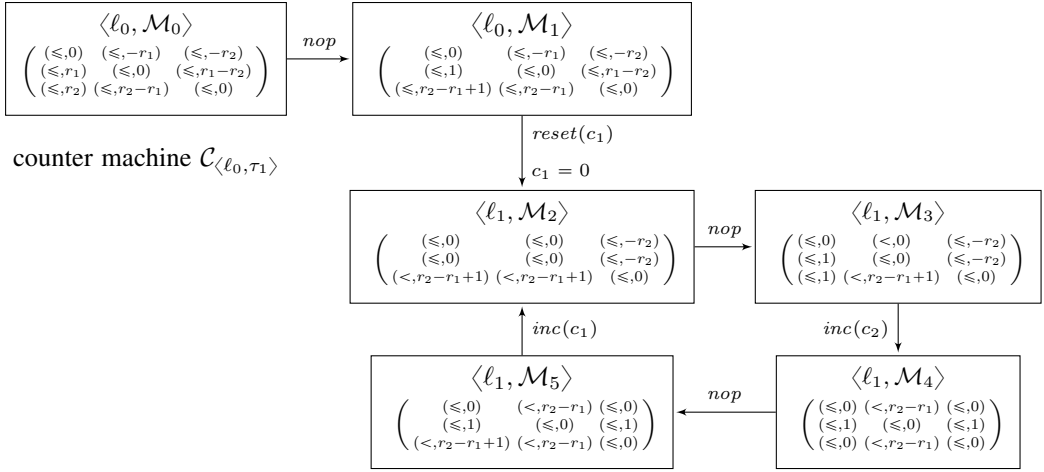
\begin{figure*}[t]
\centering
\begin{tikzpicture}[>=latex',shorten >=1pt,node distance=1.9cm,on grid,auto,
state/.style={
           rectangle,
           draw=black, 
           minimum height=1.5cm,
					minimum width=1.7cm,
           inner sep=1pt,
           text centered,
           },
roundnode/.style={circle, draw,minimum size=1.2mm},]

\node [state,label={[label distance=.3cm]-90:counter machine $\mathcal{C}_{\tuple{\loc_0,\tau_1}}$}] (l0z0) at(0,0) {\scalebox{.8}{\begin{tabular}[t]{l}  \\ 
$\vect{(\leq,0)&  (\leq,-r_1)&  (\leq,-r_2)\\
					 (\leq, r_1)&     (\leq,0)&   (\leq,r_1-r_2)\\
(\leq,r_2)&  (\leq,r_2-r_1)&     (\leq,0)}$\end{tabular}}};
\node [draw=none] (dum) [above=.45 of l0z0] {{\small $\tuple{\ell_0,\M_0}$}};

\node [state] (l0z1) [right=5cm of l0z0] {\scalebox{.8}{
	\begin{tabular}[t]{l}  \\ 
$\vect{(\leq,0)&  (\leq,-r_1)&  (\leq,-r_2)\\
					 (\leq,1)&     (\leq,0)&   (\leq,r_1-r_2)\\
(\leq,r_2-r_1+1)&  (\leq,r_2-r_1)&     (\leq,0)}$\end{tabular}}};
\node [draw=none] (dum) [above=.45 of l0z1] {$\tuple{\ell_0, \M_1}$};

\node [state] (l1z2) [below=2.5cm of l0z1] {\scalebox{.8}{
	\begin{tabular}[t]{l}  \\ 
$\vect{(\leq,0)&  (\leq,0)&  (\leq,-r_2)\\
					 (\leq,0)&     (\leq,0)&   (\leq,-r_2)\\
(<,r_2-r_1+1)&  (<,r_2-r_1+1)&     (\leq,0)}$\end{tabular}}};
\node [draw=none] (dum) [above=.45 of l1z2] {{\small $\tuple{\ell_1,\M_2}$}};

\node [state] (l1z3) [right=5cm of l1z2]  {\scalebox{.8}{
	\begin{tabular}[t]{l}  \\ 
$\vect{(\leq,0)&  (<,0)&  (\leq,-r_2)\\
					 (\leq,1)&     (\leq,0)&   (\leq,-r_2)\\
(\leq,1)&  (<,r_2-r_1+1)&     (\leq,0)}$\end{tabular}}};
\node [draw=none] (dum) [above=.45 of l1z3] {{\small $\tuple{\ell_1,\M_3}$}};

\node [state] (l1z4) [below=2.2cm of l1z3] {\scalebox{.8}{
	\begin{tabular}[t]{l}  \\ 
$\vect{(\leq,0)&  (<,r_2-r_1)&  (\leq,0)\\
					 (\leq,1)&     (\leq,0)&   (\leq,1)\\
(\leq,0)&  (<,r_2-r_1)&     (\leq,0)}$\end{tabular}}};
\node [draw=none] (dum) [above=.45 of l1z4] {{\small $\tuple{\ell_1,\M_4}$}};
\node [state] (l1z5) [below=2.2cm of l1z2]{\scalebox{.8}{
	\begin{tabular}[t]{l}  \\ 
$\vect{(\leq,0)&  (<,r_2-r_1)&  (\leq,0)\\
					 (\leq,1)&     (\leq,0)&   (\leq,1)\\
(<,r_2-r_1+1)&  (<,r_2-r_1)&     (\leq,0)}$\end{tabular}}};
\node [draw=none] (dum) [above=.45 of l1z5] {$\tuple{\ell_1,\M_5}$};
\path[->] (l0z0) edge node [above,midway] {\scriptsize{$\mathit{nop}$}}  (l0z1);
\path[->] (l0z1) edge node [near start,right] {\scriptsize{$\mathit{reset}(c_1)$}} node [right,near end] {\scriptsize{$c_1=0$}}  (l1z2);
\path[->] (l1z2) edge node [above,midway] {\scriptsize{$\mathit{nop}$}} (l1z3);
\path[->] (l1z3) edge node [left,midway] {\scriptsize{$\mathit{inc}(c_2)$}}  (l1z4);
\path[->] (l1z4) edge node [above,midway] {\scriptsize{$\mathit{nop}$}} (l1z5);
\path[->] (l1z5) edge node [right,midway] {\scriptsize{$\mathit{inc}(c_1)$}}   (l1z2);

\end{tikzpicture}
\caption{The (relevant part of the) counter machine~$\mathcal{C}_{\tuple{\loc,\tau_1}}$ 
constructed from the timed automaton in Figure~\ref{fig:reachAutomata},
 where $\tau_1$ is the type of the valuation~$\val$ with $\val_1=0.6$ and $\val_2=0$.  
The placement of a transition between~$\tuple{\loc_1, \M_5}$ and $\tuple{\loc_1, \M_2}$ relies on the
fact that terms $-r_2$ and $0$ are equivalent under the preorder induced by~$\tau_1$. 
}
\label{fig:symbolicAutomata}
\end{figure*}

We are now in a position to state our main result.
\begin{theorem}
	\label{theorem_main_ta_formula}
Given a timed automaton $\A$ with $n$ clocks and locations
$\loc,\loc'$, we can compute in exponential time a formula
\[ \varphi_{\loc,\loc'}(z_1,\ldots,z_n,r_1,, \ldots,r_n,z'_1,\ldots,z'_n,r'_1,,\ldots,r'_n) \]
in the existential fragment\footnote{We claim that this result can be strengthened to 
state that the reachability relation can be expressed by a \emph{quantifier-free} formula, again computable in exponential time.
To do this one can exploit structural properties of the class of monotonic counter machine 
that arise from timed automata.  We omit details.} of 
$\mathcal{L}_{\mathbb{R},\mathbb{Z}}$ such that there is a finite run in
$\A$ from state~$\tuple{ \ell,\val}$ to state~$\tuple{ \ell',\val'}$
just in case
\[ \tuple{\floor{\val},\mathrm{frac}(\val),\floor{\val'},\mathrm{frac}(\val')} \models
\varphi_{\loc,\loc'} \, . \]
\end{theorem}
\begin{proof}
We give the definition of $\varphi_{\loc,\loc'}$ below and justify
the complexity bound in Appendix~\ref{TimeBound}.

For simplicity we write formula $\varphi_{\loc,\loc'}$ as a
disjunction over the collection $\mathrm{Tp}_n$ of all $n$-types.
However each disjunct only depends on the restriction of the type
$\tau$ to the (finite) set of atomic formulas $t\leq t'$ with $t,t'
\in \mathcal{DT}_{\mathbb{R}}(n)$; so $\varphi_{\loc,\loc'}$ can
equivalently be written as a finite disjunction.  We define
\begin{equation}
\label{eq-formula}
 \varphi_{\loc,\loc'} \defequals \bigvee_{\tau\in\mathrm{Tp}_n} 
\alpha^\tau \wedge \chi^\tau_{\loc,\loc'} \, 
\end{equation}
where the subformulas $\alpha^\tau$ and $\chi_{\loc,\loc'}^\tau$ are
defined below.

The Hintikka formula  $\alpha^\tau(r_1,\ldots,r_n)$\footnote{Recall that by
  convention~$[r_0]=[0]$, thus we treat variable~$r_0$ as synonymous with the
  constant $0$.} is defined by
\[ \alpha^\tau  \defequals
\bigwedge_{\substack{t,t'\in\mathcal{DT}_{\mathbb{R}}(n)\\ (t\leq t') \in \tau}} t\leq t'
\wedge
\bigwedge_{\substack{t,t'\in\mathcal{DT}_{\mathbb{R}}(n)\\ (t\leq t') \not\in \tau}} \neg(t\leq t') \, .\]
Given a valuation $\val\in\RP^{\clocks}$,
$\mathrm{frac}(\val) \models \alpha^\tau$ just in case the set of
difference formulas satisfied by $\mathrm{frac}(\val)$ is identical to
the set of difference formulas in $\tau$.

Formula $\chi^\tau_{\loc,\loc'}$ is defined by writing
\begin{align*}
 \chi^\tau_{\loc,\loc'}  \defequals & 
\bigvee_{\substack{\mathcal{M}\in\closure(\mathcal{M}_\tau  )\\
\mathcal{M}=(\prec_{i,j},m_{i,j})}} 
\Big(\psi_{\tuple{\loc,\mathcal{M}_\tau},\tuple{\loc',\mathcal{M}}}
(z_1,\ldots,z_n,z'_1,\ldots,z_n') \\
& \qquad \qquad \wedge \bigwedge_{0 \leq i,j \leq n} r'_i-r'_j \prec_{i,j} m_{i,j} \Big) \, .
\end{align*}
Here the subformula
$\psi_{\tuple{\loc,\mathcal{M}_\tau},\tuple{\loc',\mathcal{M}}}$,
expresses the reachability relation in the counter machine
$\mathcal{C}_{\tuple{\loc,\tau}}$ between control states
$\tuple{\loc,\mathcal{M}_\tau}$ and $\tuple{\loc',\mathcal{M}}$, as
per Proposition~\ref{prop:monotone}.  Recall from
Corollary~\ref{corl:well-supported} that each $m_{i,j}$ is a
difference term involving variables $r_0,\ldots,r_n$.  The correctness
of $\varphi_{\loc,\loc'}$ is immediate from
Proposition~\ref{prop:monotone} and Theorem~\ref{thm:symbolic}. 
\end{proof}

\begin{example}\label{example-family}
Consider the timed automaton~$\A$ in Figure~\ref{fig:reachAutomata}. Fix the
type~$\tau_1$ for the valuation~$\vect{0.6\\0}$.  
We illustrate the
relevant part of the counter machine $\mathcal{C}_{\tuple{\loc_0,\tau_1}}$ in
Figure~\ref{fig:symbolicAutomata}.
 States~$\tuple{\ell,\M}$ of the
machine comprise a location $\ell$ and parametric
DBM~$\M$. 
Moreover, $\M_0=\M_{\tau_1}$.  The placement of a transition
between~$\tuple{\loc_1, \M_5}$ and $\tuple{\loc_1, \M_2}$ relies on
the fact that terms $-r_2$ and $0$ are equivalent with respect to the
equivalence relation on terms induced by~$\tau_1$.

Let $\alpha^{\tau_1}$ be the Hintikka formula of the type~$\tau_1$. Clearly, 
$\tuple{0.6,0}\models \alpha^{\tau_1}$.
We define $ \chi^\tau_{\loc_0,\loc_1}$  as follows:
\[\begin{aligned} 
\chi^{\tau_1}_{\loc_0,\loc_1}\defequals &(z_1=0)  \wedge
\\ & \Big[ [(z'_2-z'_1=z_2-z_1) \wedge  (\psi_2 \vee \psi_3 )] \vee \\
& [(z'_2-z'_1=-1+z_2-z_1) \wedge (\psi_4 \vee \psi_5 )]\Big],
\end{aligned}\]

where $\psi_1 $, $\psi_2 $,  $\psi_3 $ and $\psi_4$ are given in the following:
\[
\begin{aligned}
\psi_2&\equiv (r'_1 =0) \wedge (r_2\leq r'_2 <r_2-r_1+1),  \\[2mm]
\psi_3&\equiv (0< r'_1) \wedge (r_2 \leq  r'_2)\\
&\quad \wedge (r_2 \leq r'_2-r'_1 <r_2-r_1+1),\\[2mm]
\psi_4&\equiv (r_2-r_1< r'_1) \wedge (r'_2=0),   \\[2mm]
\psi_5&\equiv (r_2-r_1 < r'_1) \wedge (r'_2< r_2-r_1+1) \\
&\quad \wedge (-1 \le r_2'-r'_1<r_2-r_1). \\
\end{aligned}
\]

The formulae~$\psi_i$ (with $i\in \{2,3,4,5\}$)
summarise the constraints placed on~$r'_1$ and $r'_2$ by the
parametric DBMs~$\M_i$ in the counter machine~$\mathcal{C}_{\tuple{\loc_0,\tau_1}}$.
 See Figure~\ref{fig:symbolicAutomata} for the 
given constraints in the parametric DBMs~$\M_i$.
Recall that real-valued variables~$r_i,r'_i$  range over the interval~$[0,1]$.

Let~$\tau_2$ be the type for the valuation~$\vect{0\\0.2}$.  In
comparison with~$\mathcal{C}_{\tuple{\loc_0,\tau_1}}$, we present the
counter machine $\mathcal{C}_{\tuple{\loc_0,\tau_2}}$ in
Figure~\ref{fig:symbolicAutomata2} in Appendix~\ref{symbolicautmatafigure}. 

The formula~$\varphi_{\loc_0,\loc_1}$, expressing the set of
valuations~$\val$ and $\val'$ such that $\tuple{\loc_1,\val'}$ is
reachable from~$\tuple{\loc_0,\val}$, is then the disjunction of all
formulas~$\alpha^{\tau} \wedge \chi^\tau_{\loc_0,\loc_1}$ for types~$\tau\in\mathrm{Tp}_n$:
\[\varphi_{\loc_0,\loc_1}=(\alpha^{\tau_1} \wedge \chi^{\tau_1}_{\loc_0,\loc_1}) \vee (\alpha^{\tau_2} \wedge \chi^{\tau_2}_{\loc_0,\loc_1}) \vee\cdots .\]

\end{example}

\section{Parametric Timed Reachability Logic} 
\label{sec:MC}
Let $\A=\tuple{\locs,\clocks,\edges}$ be a timed automaton augmented
with a labelling function~$\lbl:\locs\rightarrow 2^{\prop}$.  Let
$\varphi$ be a sentence of {\ptctl}.  Recall that the model checking
problem of~$\A$ against~$\varphi$ asks, given a
state~$\tuple{\loc,\val}$ of~$\A$, whether
$\tuple{\loc,\val} \models \varphi$.

In this section we prove the following result.
\begin{theorem}\label{main-theo-NEXPTIME}
  The model-checking problem for {\ptctl} is decidable in EXPSPACE and
  is NEXPTIME-hard.
\end{theorem}

For membership in EXPSPACE, given a timed automaton~$\A$, a
configuration $\tuple{\ell,\val}$ of $\A$, and a sentence $\psi$ of
{\ptctl}, we construct in exponential time a sentence
$\widetilde{\psi}$ of $\mathcal{L}_{\mathbb{R},\mathbb{Z}}^*$ that is
true if and only if $\tuple{\ell,\val}\models \psi$.  We thereby
obtain an exponential space algorithm for the model checking problem.
We then prove NEXPTIME-hardness by a reduction from SUCCINCT 3-SAT.
 
\subsection{Reduction of Model Checking to Satisfiability}
The model checking procedure for {\ptctl} relies on a ``cut-down''
version of Theorem~\ref{theorem_main_ta_formula}, concerning the logical definability
of the reachability relation.  In this version, given as
Lemma~\ref{lem:cut-down} below, we do not represent the full
reachability relation, but instead abstract the integer parts of all
clocks except the reference clock $x_n$. 
This  abstraction is sufficient for model-checking {\ptctl}, and moreover
 allows us to
obtain a formula that lies in the sub-logic
$\mathcal{L}_{\mathbb{R},\mathbb{Z}}^*$, which has better complexity
bounds than the full logic $\mathcal{L}_{\mathbb{R},\mathbb{Z}}$.

Given $N\in\mathbb{N}$, define the set $\mathcal{R}_N$ of
\emph{regions} to be $\mathcal{R}_N= \{0,\ldots,N\} \cup \{\infty\}$.
A counter valuation $\upsilon \in \mathbb{N}^n$ is abstracted to 
$Reg(\upsilon) \in \mathcal{R}_N^n$, where
\[ Reg(\upsilon)_i = \left\{ \begin{array}{ll} \upsilon_i & \mbox{ if $\upsilon_i \leq N$}\\
\infty & \mbox{ otherwise} \end{array}\right . \] 

The following lemma is proved in Appendix~\ref{append:monotone}.
\begin{lemma}
\label{lem:cut-down}
  Let $\A$ be a timed automaton with $n$ clocks and maximum clock
  constant $N$.  Given two locations
$\loc,\loc'$ of $\A$ and $R,R' \in \mathcal{R}_N^n$, we can compute
in exponential time a quantifier-free
$\mathcal{L}_{\mathbb{R},\mathbb{Z}}^*$-formula 
\[ \varphi_{\loc,R,\loc',R'}(z,r_1,\ldots,r_n,z',r'_1,\ldots,r_n')
\]
such that  there is a finite run in
$\A$ from state~$\tuple{ \ell,\val}$ to state~$\tuple{ \ell',\val'}$,
where $Reg(\floor{\val})=R$ and $Reg(\floor{\val'})=R'$ , just in case
\[ \tuple{\floor{\val_n},\mathrm{frac}(\val),\floor{\val'_n},\mathrm{frac}(\val')} \models
\varphi_{\loc,R,\loc',R'} \, . \]
\end{lemma}

Let $\psi$ be a formula of {\ptctl} of the first type, involving the set
of parameters $\theta_1,\ldots,\theta_k$, and let $\A$ be a timed
automaton with $n$ clocks and maximum clock constant $N$.  For each
location $\loc$ of $\A$ and $R \in \mathcal{R}^n_{N}$ such that $R_n=0$, we
obtain a $\mathcal{L}_{\mathbb{R},\mathbb{Z}}^*$-formula 
\[\widetilde{\psi}_{\loc,R}(r_1,\ldots,r_n,w_1,\ldots,w_k,s_1,\ldots,s_k)\]
in real variables $\boldsymbol{r}=(r_1,\ldots,r_n)$ and
$\boldsymbol{s}=(s_1,\ldots,s_k)$ and integer variables
$\boldsymbol{w}=(w_1,\ldots,w_k)$  
such that 
\[ \tuple{\mathrm{frac}(\val),\floor{\xi},\mathrm{frac}(\xi)} 
\models \widetilde{\psi}_{\loc,R}
  \quad\mbox{iff}\quad \tuple{\loc,\val} \models_{\xi} \psi \] for all
parameter valuations~$\xi \in \RP^k$ and all clock
valuations~$\val \in \RP^n$ such that $Reg(\floor{\val})=R$ and $\val_n=0$.

To keep things simple, we assume that every configuration of $\A$ can
generate an infinite non-zeno run.  It is not difficult to drop this
assumption since the collection of configurations from which there
exists such a run is a union of clock regions and hence is definable
in $\mathcal{L}_{\mathbb{R},\mathbb{Z}}^*$.  We also assume, without
loss of generality, that the reference clock~$x_n$ is not mentioned in
any guard of~$\A$.

The construction of $\widetilde{\psi}_{\loc,R}$ is by induction on the
structure of $\psi$.  The induction cases for the Boolean connectives
are straightforward and we concentrate on the induction step for the
connective $\exists\F_{\sim \param}$.  In fact we only consider the
case that $\sim$ is the equality relation $=$, the cases for $<$ and
$>$ being very similar.

Suppose that $\psi \equiv \exists\F_{= \param_i}\psi'$ for some
{\ptctl}-formula $\psi'$ and $i\in\{1,\ldots,k\}$.  Then we define
\begin{align*}
\widetilde{\psi}_{\loc,R}& (\boldsymbol{r},\boldsymbol{w},\boldsymbol{s})
\defequals \bigvee_{\loc',R'} 
\exists \boldsymbol{r}' \exists z'
\, \varphi_{\loc,R,\loc',R'}(0,\boldsymbol{r},z',\boldsymbol{r}') \\
& \wedge (r_n'=s_i \wedge z'=w_i) \wedge \widetilde{\psi'}_{\loc',R'}(r_1'\ldots,r_{n-1}',0,\boldsymbol{w},\boldsymbol{s})
\end{align*}
where $\varphi_{\loc,R,\loc',R'}$ is the reachability formula defined
in Lemma~\ref{lem:cut-down}.  Note that this definition relies on the assumption
that the clock $x_n$ is never reset by the timed automaton and hence
can be used to keep track of global time.
  
This completes the translation of {\ptctl}-formulas of the first type
to formulas of $\mathcal{L}_{\mathbb{R},\mathbb{Z}}^*$.  Extending
this inductive translation to {\ptctl}-formulas of the second type is
straightforward, bearing in mind that we represent each parameter
$\theta_i$ by a variable $w_i$ for its integer part and a variable
$s_i$ for its fractional part.  Thus, e.g., the {\ptctl}-formula
$\exists \theta_i \psi$ is translated as $\exists w_i \exists s_i (0
\leq s_i<1 \wedge \widetilde{\psi})$.

Given a sentence $\psi$ of {\ptctl}, location $\loc$ of $\A$, and $R
\in \mathcal{R}_N$, our translation yields a formula
$\widetilde{\psi}_{\loc,R}(r_1,\ldots,r_n)$ such that for any
valuation $\val$ with $Reg(\floor{\val})=R$ we have $\tuple{\loc,\val}\models
\psi$ if and only if $\mathrm{frac}(\val) \models
\widetilde{\psi}_{\loc,R}$.  By Lemma~\ref{lem:cut-down}, formula
$\widetilde{\psi}_{\loc,R}$ has size singly exponential in the size of
$\psi$ and $\A$ and quantifier-depth linear in the size of $\psi$.

The model checking problem then reduces to determining the truth of
$\widetilde{\psi}_{\loc,R}$ on $\mathrm{frac}(\val)$, where
$Reg(\floor{\val})=R$.  Since satisfiability for sentences of
$\mathcal{L}_{\mathbb{R},\mathbb{Z}}^*$ can be decided in polynomial
space in the formula size and exponential space in the number of
quantifiers (by Proposition~\ref{prop:diff-bound}), the model checking
problem of {\ptctl} lies in EXPSPACE.

\subsection{NEXPTIME-Hardness}
\label{sec:nexptimehard}
In this section we show that model checking timed automata against the
fixed {\ptctl} sentence $\exists \param \, \forall\hspace{-1.3mm}\G_{=
  \param} p$ is NEXPTIME-hard.  We remark that, due to the punctual
constraint ${=\!\theta}$, the above formula expresses a
synchronization property---\emph{there exists a duration $\theta$ such
  that all runs are in a $p$-state after time exactly $\theta$}.

Recall that a \emph{Boolean circuit} is a finite directed acyclic
graph, whose nodes are called \emph{gates}. An \emph{input gate} is a
node with indegree $0$. All other gates have label either $\vee$,
$\wedge$, or $\neg$. An \emph{output gate} is a node with outdegree
$0$.

We show {NEXPTIME}-hardness by reduction from the SUCCINCT 3-SAT
problem.  The input of SUCCINCT 3-SAT is a Boolean circuit~$C$,
representing a 3-CNF formula $\varphi_C$, and the output is whether or
not $\varphi_C$ is satisfiable.  Specifically, $C$ has $2$ output
gates, and the input gates are partitioned into two nonempty sets of
respective cardinalities $n$ and $m$.  The formula~$\varphi_C$ has
$2^n$ variables and $2^m$ clauses (in particular, the number of
variables and clauses in $\varphi_C$ can be exponential in the size of~$C$).  
The first~$n$ inputs of~$C$ represent the binary encoding of
the index $i$ of a variable, and the remaining $m$ inputs of $C$
represent the binary encoding of the index $j$ of a clause in
$\varphi_C$.  The output of $C$ indicates whether the $i$-th variable
occurs positively, negatively, or not at all in the $j$-th clause of
$\varphi_C$.  The SUCCINCT 3-SAT problem is
{NEXPTIME}-complete~\cite{PapadimitriouY86}.

Given an instance of SUCCINCT 3-SAT, that is, a Boolean circuit $C$ as
described above, we construct a timed automaton $\A$ augmented with a
labelling function $\lbl$ such that the 3-CNF formula $\varphi_C$
encoded by circuit $C$ is satisfiable if and only if
$(\loc,\zeroval)\models \exists \param \forall \G_{= \param} p$ for
some designated location $\loc$.  

There are two ideas behind the reduction.  First we construct a linear
bounded automaton $\mathcal{B}$ from the circuit $C$ such that,
roughly speaking, the 3-CNF formula $\varphi_C$ is satisfiable if and
only if there exists an integer $N$ such that, starting from an
initial configuration, all length-$N$ paths in the configuration graph
of $\mathcal{B}$ end in a configuration with label $p$.  The second
part of the reduction is to simulate encode the configuration graph of
$\mathcal{B}$ as the configuration graph of a timed automaton $\A$.

We construct $\mathcal{B}$ such that its number of control states is
polynomial in the size of $C$, and we fix an initial tape
configuration of $\mathcal{B}$ of length likewise bounded by a
polynomial in the size of $C$.  We designate certain transitions of
$\mathcal{B}$ as $\checkmark$-transitions.  In every computation of
$\mathcal{B}$, the sequence of steps between the $i$-th and $(i+1)$-st
$\checkmark$-transitions, for $i\in\mathbb{N}$, is referred to as the
\emph{$i$-th phase} of the computation.  We design $\mathcal{B}$ so
that the number of steps in the $i$-th phase is independent of the
nondeterministic choices along the run.

The definition of $\mathcal{B}$ is predicated on a numerical encoding
of propositional valuations.  Suppose that $X_1,\ldots,X_{2^n}$ are
the variables occurring in $\varphi_C$, and write $p_1,\ldots,p_{2^n}$
for the first $2^n$ prime numbers in increasing order.  Given a
positive integer $N$, we obtain a Boolean valuation of
$X_1,\ldots,X_{2^n}$ in which $X_j$ is false if, and only if,
$N \bmod p_j = 0$.  With this encoding in hand, we proceed to define
$\mathcal{B}$:

\begin{enumerate}
\item In the first phase, $\mathcal{B}$ guesses three $n$-bit numbers
$1 \leq i_1,i_2,i_3 \leq 2^n$ and a single $m$-bit number
$1 \leq j \leq 2^m$ and writes them on its tape.
\item In the second phase, $\mathcal{B}$ computes the three prime numbers
  $p_{i_1},p_{i_2},p_{i_3}$ and writes them on its tape.
\item In the third phase, by simulating the circuit $C$, $\mathcal{B}$
  determines whether the propositional variables
  $X_{i_1},X_{i_2},X_{i_3}$ appear in the $j$-th clause of $\varphi_C$, henceforth denoted
$\psi_{j}$. 
  If one of them does not appear at all, then $\mathcal{B}$ moves into an accepting self-loop. 
  Otherwise, $\mathcal{B}$ remembers in its state whether $X_{i_1}, X_{i_2}, X_{i_3}$ appear positively or negatively in $\psi_{j}$, 
  and then 
$\mathcal{B}$ proceeds to the next phase.
\item From phase four onwards, $\mathcal{B}$ maintains on its tape
  three counters, respectively counting modulo
  $p_{i_1},p_{i_2},p_{i_3}$.  In every successive phase, each of these
  counters is incremented by one.  At the end of each phase, 
  $\mathcal{B}$ checks whether the values of the counters encode a satisfying valuation of clause
  $\psi_j$. If this is the case, then $\mathcal{B}$ moves into an accepting state. Otherwise $\mathcal{B}$ proceeds to the next phase. 
\end{enumerate}
By construction, $N\in\mathbb{N}$ encodes a satisfying valuation of
$\varphi_C$ if and only if all computation paths of $\mathcal{B}$
reach an accepting state at the end of the $(N+3)$-rd phase. 

It remains to explain how from $\mathcal{B}$ one can define a timed
automaton $\A$ whose configuration graph embeds the configuration
graph of $\mathcal{B}$.  The construction is adapted from the
PSPACE-hardness proof for reachability in timed
automata~\cite{AlurD94}.  We refer to Appendix~\ref{NEXPTIME-hard} for
details of this construction.  In the end, the initial configuration
$(\ell,\boldsymbol{0})$ of $\A$ satisfies
$\exists \param `, \forall \G_{= \param} p$ if and only if $\varphi_C$
is satisfiable.

\section{Conclusion}
We  have given a new proof of the result of Comon and Jurski that
the reachability relation of a timed automaton is definable in linear
arithmetic.  In addition to making the result more accessible, our
main motivations in revisiting this result concerned potential
applications and generalisations.  With regard to applications, we
have already put the new proof to work in deriving complexity bounds
for model checking the reachability fragment of parametric {\tctl}.  In
future work we would like to see whether ideas from this paper can be
applied to give a more fine-grained analysis of extensions of timed
automata, such as timed games and priced timed automata.

We claim that a finer analysis of the complexity of our decision
procedure for model checking {\ptctl} yields membership of the problem
in the complexity class $\mathrm{STA}(*,2^{O(n)},n)$, i.e., the class
of languages accepted by alternating Turing machines running in time
$2^{O(n)}$ and making at most $n$ alternations on an input of length
$n$.  This improved upper bound follows from a refinement of the
statement of Proposition~\ref{prop:diff-bound}, on the complexity of
the decision problem for $\mathcal{L}^*_{\mathbb{R},\mathbb{Z}}$, to
state that the truth of a prenex-form sentences of size $n$ and with
$k$ quantifier alternations can decided by a polynomial time
alternating Turing machine, making at most $k$ alternations.

We claim also that our NEXPTIME-hardness result can be strengthened to
match the new upper bound.  The idea here would be to reduce a version
of SUCCINCT 3-SAT with quantifier alternation to model
checking $\ptctl$ formulas of the form $Q_1\theta_1 \ldots
Q_k \theta_k \forall\G_{=\param_1} \ldots \forall\G_{=\param_k} p$ for
$Q_1,\ldots,Q_k$ a sequence of quantifiers with $k$ alternations.

Details of the improved upper and lower complexity bound will appear in
a subsequent version of this paper.

{\smallskip\bf\noindent Acknowledgements.}
This work was partially supported by the EPSRC
through grants  EP/M012298/1 and EP/M003795/1 and 
by the German Research Foundation (DFG), project QU 316/1-2.

\bibliography{IEEEabrv,tctl_lit}

\newpage
\onecolumn
\appendix

\subsection{Proof of Proposition~\ref{prop:diff-bound}}\label{append-prop-diff-bound}

We first recall that the language
$\mathcal{L}^*_{\mathbb{R},\mathbb{Z}}$ has terms   of both
real-number sort and integer sort, where the atomic formulas:
\begin{itemize}
	\item if  integer sort,  have form
\begin{gather}
 z-z' \leq c \, \mid \, z \leq c \, \mid \,
   {z-z' \equiv c} \pmod d 
\label{eq:diff}
\end{gather}
for integer variables $z,z'$ and integers $c,d$.
\item if real-number sort, have form~$t\leq t'$ where $t,t'$ 
are derived by the grammar
\[ t::= c \, \mid \, r \mid \, t+t \, \mid t-t \, , \] where
$c\in \mathbb{Q}$ is a constant and $r\in\{r_0,r_1,\ldots\}$ is a
real-valued variable.
\end{itemize}

One can prove Proposition~\ref{prop:diff-bound} by combining the
quantifier-elimination procedures of Ferante and
Rackoff~\cite{FerranteR75,kozenBook} for
$\mathcal{L}_{\mathbb{R}}$ and To~\cite[Section 4]{To09} for the
fragment of Presburger arithmetic in which atomic formulas have the
form shown in (\ref{eq:diff}).

To eliminate quantifiers in formulas of real arithmetic, Ferante and
Rackoff~\cite{FerranteR75} define an
equivalence relation~$\rrel{k}{m}$ on $k$-tuples of real numbers.  The
relation is such that $\rrel{k}{m}$-equivalent $k$-tuples agree on all
quantifier-free formulas in which all constants have the largest (absolute) constant at most
$m$.  We refer the reader to~\cite{kozenBook}
for the definition of $\rrel{k}{m}$; here we just recall the key
results.

Let $A^k_m$ be the set of all affine functions~$f: \R^k \to \R$ with integer coefficients, 
where all constants and coefficients have the largest (absolute) constant at most
$m$.

\begin{lemma}[Lemma 22.3 and 22.4 from~\cite{kozenBook}]\label{lem:22.3}
Given two $k$-tuples $\boldsymbol{a}=(a_1,\cdots,a_k)$ and
$\boldsymbol{b}=(b_1,\cdots,b_k)$ of real numbers such that 
$\boldsymbol{a}\rrel{k}{2m^2} \boldsymbol{b}$ for some~$m\in\mathbb{Z}_{>0}$,
then for all $c \in \R$
there exists $d \in \R$ such that $(\boldsymbol{a},c) \rrel{k+1}{m} (\boldsymbol{b},d)$.
Moreover, $d$ can be chosen to have the form~$f(\boldsymbol{b})/e$ 
where $f \in A^k_{2m^2}$ and $\abs{e}\leq 2m^2$.
\end{lemma}

To eliminate quantifiers in formulas of the above fragment of Presburger arithmetic, 
analogue to the relation~${\rrel{k}{m}}$, To~\cite[Definition 6]{To09} has defined an equivalence relation~$\zrel{k}{p,m}$
on $k$-tuples of integers, where $p,m \in \mathbb{Z}_{>0}$. 
The relation is such that~${\zrel{k}{p,m}}$-equivalent $k$-tuples agree on all
quantifier-free formulas, where  all constants have the largest (absolute) constant at most
$m$ and the period of the formula is~$p$. 
The period of the formula is  the least common multiple of the 
periods~$e$ of each atomic term ${z-z' \equiv c} \pmod e$.
 We refer the reader to~\cite{To09}
 for the definition of $\zrel{k}{p,m}$; here we just recall the key
results.

\begin{lemma}[Lemma 7 and 8 from~\cite{To09}]
Given two $(k+1)$-tuples $\boldsymbol{a}=(a_0,\cdots,a_k)$ and $\boldsymbol{b}=(b_0,\cdots,b_k)$ of integers 
such that $a_0=b_0=0$ and $\boldsymbol{a} \zrel{k}{p,3m} \boldsymbol{b}$ for some~$p,m>0$,
then for all
$c \in \N$ there exists $b \in \N$ such that $(\boldsymbol{a},c) \zrel{k+1}{p,m} (\boldsymbol{b},d)$.
Moreover, $d$ can be chosen such that $0\leq d\leq max(b_0,\cdots,b_k)+pm+p$.
\end{lemma}

Fix $m \in \mathbb{Z}_{>0}$. For all $n\in \N$, define $g(0,m)\defequals m$ and $g(n+1,m)\defequals 2 g(n,m)^2$, 
moreover, define $h(0,m)\defequals m$ and $h(n+1,m)\defequals 3 h(n,m)$.

\begin{lemma}
\label{lem:big-lemma}
Let $\varphi(r_1,\cdots,r_k,z_0, \cdots,z_{k'})$ be a formula in~$\mathcal{L}^*_{\mathbb{R},\mathbb{Z}}$, 
with $k$ free real-valued variables~$r_i$ and $k'$ free integer variables~$z_i$,
and with $n$ quantifiers over real-valued variables and $n'$ quantifiers over integer-valued variables,
where $m$ is the largest (absolute) constant and~$p$ is the period of the formula.
Suppose $\boldsymbol{a},\boldsymbol{b} \in \R^k$ are $k$-tuples of real numbers 
such that $\boldsymbol{a}\rrel{k}{g(n,m)} \boldsymbol{b}$.  
Suppose $\boldsymbol{a'},\boldsymbol{b'} \in \N^{k'}$ are $k'$-tuples of integers
such that $\boldsymbol{a'}\zrel{k'}{p,h(n',m)} \boldsymbol{b'}$. 
Then,  we have 	
\[ \varphi(\boldsymbol{a},\boldsymbol{a'}) \text{ holds} \Leftrightarrow \varphi(\boldsymbol{b},\boldsymbol{b'}) \text { holds.}\] 
\end{lemma}

\begin{proof}
The proof is by an induction on the structure of the formula. 
For atomic formulas, each sort, the result is immediate from the definition of
equivalence relations  $\rrel{k}{m}$ and $\zrel{k}{p,m}$. 
For the Boolean connectives, the result is straightforward using the induction hypothesis for each subformula.

\begin{itemize}
	\item For formulas $\exists r \, \varphi(\boldsymbol{a},\boldsymbol{a'},r)$, where $r$ is a real-valued variable: 
	suppose $\exists r \,\varphi(\boldsymbol{a},\boldsymbol{a'},r)$ holds, and let 
	$c\in R$ be such that $\varphi(\boldsymbol{a},\boldsymbol{a'},c)$ holds. Since $\boldsymbol{a}\rrel{k}{g(n,m)} \boldsymbol{b}$, by Lemma~\ref{lem:22.3}, 
	there exists some $d$ such that $(\boldsymbol{a},c)\rrel{k+1}{g(n-1,m)} (\boldsymbol{b},d)$. 
	Applying the induction hypothesis, $\exists r \, \varphi(\boldsymbol{b},\boldsymbol{b'},r)$ holds, too.
	
	\item  For formulas $\exists z \, \varphi(\boldsymbol{a},\boldsymbol{a'},z)$, where $z$  is a integer-valued variable: 
	suppose $\exists z \, \varphi(\boldsymbol{a},\boldsymbol{a'},z)$ holds, and let 
	$c'\in R$ be such that $\varphi(\boldsymbol{a},\boldsymbol{a'},c')$ holds. Since $\boldsymbol{a'}\zrel{k}{p,h(n',m)} \boldsymbol{b'}$, by Lemma~\ref{lem:22.3}, 
	there exists some $d'$ such that $(\boldsymbol{a'},c')\zrel{k+1}{p,h(n'-1,m)} (\boldsymbol{b'},d')$. 
	Applying the induction hypothesis, $\exists z \,\varphi(\boldsymbol{b},b',z)$ holds, too. 
\end{itemize}
 
The step of induction  for  formulas $\forall r \,\varphi(\boldsymbol{a},a',r)$ 
and $\forall z \,\varphi(\boldsymbol{a},a',z)$ are  similar.
\end{proof}

Given a prenex-form sentence $\varphi$ of
$\mathcal{L}^*_{\mathbb{R},\mathbb{Z}}$, using
Lemma~\ref{lem:big-lemma} we derive an equivalent formula in which all
quantifiers range over finite domains.  Specifically, if $\varphi$ has
$n$ quantifiers over real variables and $n'$ quantifiers over integer
variables, maximum constant $m$, and period $p$, then the real-valued
quantifiers can be restricted to range over rationals whose numerator
and denominator is at most $g(n,m) =2^{2^{n}-1}m^{2^n}= 2^{2^{O(n+\log
    \log m)}}$ and the integer quantifiers can be restricted to range
over numbers of the largest (absolute) constant at most $p(n'+1)h(n',m)+p(n'+1)=
p(n'+1)3^{n'}(m+1)=2^{O(n'+\log m+\log p )}$.  Thus the truth of
$\varphi$ can be established by an alternating Turing machine using
space exponential in $n+n'$ and polynomial in the size of the
quantifier-free part of $\varphi$.  This concludes the proof of
Proposition~\ref{prop:diff-bound}.

\subsection{Proof of Lemma~\ref{lem:support}}\label{append-lem-support}
In this section we prove Lemma~\ref{lem:support} from Subsection~\ref{sec:closure}.

Recall that DBMs have entries in $\mathbb{V}=(\{<,\leq\} \times \mathbb{R}) \cup \{(<,\infty)\}$.  In this section we
denote the order $\leq_{\mathbb{V}}$ simply by $\leq$ (and the
corresponding strict order by $<$).  Recall that a DBM
is \emph{atomic} if all but at most one entry is the trivial constraint
$(<,\infty)$.  Recall also that DBM~$M$ is \emph{consistent} if 
$(\leq,0) \leq M_{i,i}$ for all $0 \leq i \leq n$.  Write
$\mathbb{Z}_{\infty}$ for $\mathbb{Z}\cup\{\infty\}$.

\subsubsection{Tightness}  
In order to prove Lemma~\ref{lem:support}, we first introduce the
concept of \emph{tightness} for DBMs and prove that, for a clock
valuation $\val \in [0,1]^n$, every DBM in $\closure(M_\val)$
is tight.

Let $M$ be a DBM of dimension $(n+1) \times (n+1)$.  We say 
that $M$ is \emph{tight} if $M_{i,j} = M_{i,n}+M_{n,j}$ for every
pair of indices $i,j$ with $m_{i,j}\not\in\mathbb{Z}_{\infty}$.

\begin{proposition}
If $M$ is tight, then $\overrightarrow{M}$ is tight.
\label{prop:tight1}
\end{proposition}
\begin{proof}
	Write $M'=\overrightarrow{M}$ and assume that $m'_{i,j}\not\in\mathbb{Z}_\infty$ for some $0\leq i,j \leq n$. 
We show that 
$M'_{i,j}=M'_{i,n}+M'_{n,j}$.
Indeed, since $m'_{i,j}\not\in\mathbb{Z}_\infty$ we have
$j\neq 0$ and thus
\begin{eqnarray*}
M'_{i,j} & = & M_{i,j}\\
& = & M_{i,n} + M_{n,j} \quad\mbox{ ($M$ is tight)}\\
& = & M'_{i,n} + M'_{n,j}\quad\mbox{ (since $n,j\neq 0$).}
\end{eqnarray*}
\end{proof}

\begin{proposition}
Suppose that $M$ is tight and $M'$ is atomic. 
Then $M''=M\cap M'$ is tight.
\label{prop:tight3}
\end{proposition}
\begin{proof}
Suppose that $m''_{i,j}\not\in\mathbb{Z}_\infty$. 
We show that $M''_{i,j}=M''_{i,n}+M''_{n,j}$.  
There are two main cases.
First suppose that $M''_{i,j}=M_{i,j}$.  Then 
\begin{eqnarray*}
M''_{i,j} & \leq  & M''_{i,n}+M''_{n,j} \quad\mbox{ ($M''$ canonical)}\\
        & \leq & M_{i,n}+M_{n,j} \quad\mbox{ ($M'' \leq M$ pointwise)}\\
        & = & M_{i,j} \quad\mbox{ ($m_{i,j}\not\in\mathbb{Z}_\infty$, $M$ tight)} \, .
\end{eqnarray*}
Since we assume that $M''_{i,j}=M_{i,j}$, all the inequalities above are tight and we conclude that
$M''_{i,j}=M''_{i,n}+M''_{n,j}$.  

The second case is that $M''_{i,j} < M_{i,j}$.  Then by
definition of $M''$ we have $M''_{i,j} = M_{i,p}+M'_{p,q}+M_{q,j}$.
Since $m''_{i,j}\not\in\mathbb{Z}_\infty$, we must have either
$m_{i,p}\not\in\mathbb{Z}_\infty$ or
$m_{q,j}\not\in\mathbb{Z}_\infty$.
We will handle the first of
these two subcases; the second follows by symmetric reasoning.

If $m_{i,p}\not\in\mathbb{Z}_\infty$ then
\begin{eqnarray*}
M''_{i,j} &=& M_{i,p}+M'_{p,q}+M_{q,j} \quad\mbox{ (definition of $M''$)}\\
          &=& M_{i,n}+M_{n,p}+M'_{p,q}+M_{q,j} \mbox{ ($m_{i,p}\not\in\mathbb{Z}_\infty$, $M$ tight)}\\
          &\geq & M''_{i,n}+M''_{n,p}+M''_{p,q}+M''_{q,j}\quad\mbox{ ($M''\leq M,M'$ pointwise)} \\
          &\geq & M''_{i,n}+M''_{n,j} \quad\mbox{ ($M''$ canonical)}
\end{eqnarray*}
But $M''_{i,j} \leq M''_{i,n}+M''_{n,j}$ by canonicity of $M''$.  Hence 
$M''_{i,j} = M''_{i,n}+M''_{n,j}$.
\end{proof}

\begin{proposition}
Suppose that $M$ is tight. 
\begin{enumerate}
\item If $\ell\neq n$ then $M[x_\ell \leftarrow 0]$ is tight.
\item $(M \cap (x_n=1))[x_n \leftarrow 0]$ is tight.
\end{enumerate}
\label{prop:tight2}
\end{proposition}
\begin{proof}
\begin{enumerate}
\item Write $M'=M[x_\ell \leftarrow 0]$, where $\ell\neq n$, and assume
  that $m'_{i,j}\not\in\mathbb{Z}_\infty$.  We show that
  $M'_{i,j}=M'_{i,n}+M'_{n,j}$.

Indeed we have
\begin{eqnarray*}
M'_{i,j} & = & M_{i_\ell,j_\ell} \quad\mbox{ (definition of $M'$)}\\
& = & M_{i_\ell,n}+M_{n,j_\ell} \quad\mbox{ ($M$ is tight, $m_{i_\ell,j_\ell}=m'_{i,j}\not\in\mathbb{Z}_\infty$)}\\
         & = & M_{i_\ell,n_\ell}+M_{n_\ell,j_\ell} \quad\mbox{ ($n_\ell=n$)}\\
         & = & M'_{i,n}+M'_{n,j} \quad\mbox{ (definition of $M'$)}.
\end{eqnarray*}
\item Write $M'=M \cap (x_n=1)$.
We know from
  Proposition~\ref{prop:tight3} that $M'$ is tight.  Moreover we have
  $M'_{n,0}=(\leq, 1)$ and $M'_{0,n}=(\leq,-1)$.  Now write
  $M''=M'[x_n\leftarrow 0]$ and assume that
  $m''_{i,j}\not\in\mathbb{Z}_\infty$.  We show that
  $M''_{i,j}=M''_{i,n}+M''_{n,j}$.  The equality is trivial if $i=n$
  or $j=n$, so we may suppose that $i,j\neq n$.

Then we have
\begin{eqnarray*}
M''_{i,j} & = & M'_{i,j} \quad\mbox{ (definition of $M''$ and $i,j\neq n$)}\\
& = & M'_{i,n}+M'_{n,j} \quad\mbox{ ($M'$ is tight, $m'_{i,j}\not\in\mathbb{Z}_\infty$)}\\
&=& M'_{i,n}+M'_{n,0} + M'_{0,n} + M'_{n,j} \quad\mbox{ ($M'_{n,0}=(\leq,1)$ and $M'_{0,n}=(\leq,-1)$)}\\
         & = & M'_{i,0}+M'_{0,j} \quad\mbox{ ($M$ tight)}\\
         & = & M''_{i,n}+M''_{n,j} \quad\mbox{ (definition of $M''$)}.
\end{eqnarray*}
\end{enumerate}
\end{proof}

\begin{proposition}
Let $\val \in [0,1]^n$ be a valuation.  Then 
every DBM $M \in \closure(M_{\val})$ is tight.
\label{prop:tight}
\end{proposition}
\begin{proof}
$M_\val$ is obviously tight.  Then by induction, using
Propositions~\ref{prop:tight1}, \ref{prop:tight3},
and \ref{prop:tight2}, every DBM in
$\closure(M_\val)$ is tight.
\end{proof}

\subsubsection{DBM Operators Preserve Well-Supportedness} 

\begin{proof}[Proof of Lemma~\ref{lem:support}]
Assume that  $\val \in [0,1]^n$ is a clock valuation. 
We prove that all consistent DBMs $M \in\closure(M_\val)$ are well-supported.
To this end, define
\[ \mathrm{Supp}_{\val}
= \{c+\val_{i}-\val_{j}\mid c\in \mathbb{Z}, 0\leq i,j\leq
n\} \cup \{ \infty \} \, . \] It suffices to show that every
consistent DBM in $\closure(M_\val)$ has entries in
$\mathrm{Supp}_{\val}$.  Indeed we have already noted that all consistent DBMs in $\closure(M_\val)$ 
are $1$-bounded; but an entry of $\mathrm{Supp}_{\val}$ lies in the
interval $[-1,1]$ only if it has the form $c+\val_i-\val_j$ for
$c\in\{-1,0,1\}$ and $0\leq i,j \leq n$.  

We prove that every consistent DBM in $\closure(M_\val)$ has entries
in $\mathrm{Supp}_{\val}$ by induction on the sequence of operations
producing such a DBM.
	
{\bf Base case.} The DBM  $M_\val$ is obviously well-supported, since 
its $(i,j)$-th entry is $\val_i-\val_j \in \mathrm{Supp}_{\val}$ for 
all $0 \leq i,j \leq n$.

{\bf Induction step.} Let $M(\prec_{i,j},m_{i,j}) \in\closure(M_\val)$
be a DBM and assume that each entry $m_{i,j}$ lies in
$\mathrm{Supp}_\val$.  We prove that all entries of the DBMs
$\overrightarrow{M}\cap \bigcap_{i=1}^n (x_i \leq 1)$,
$M[x_\ell\leftarrow 0]$, and $M\cap M'$, for $M'$ atomic, also lie in
$\mathrm{Supp}_\val$ provided that these DBMs are consistent.

It is clear that each entry of $M[x_\ell\leftarrow 0]$ lies in
$\mathrm{Supp}_\val$ since reset only permutes the entries of a DBM
and introduces $0$ as a new entry.  Likewise it is clear that each entry
of $\overrightarrow{M}$ also lies in $\mathrm{Supp}_\val$.  Thus to
complete the inductive argument it suffices to show that for any
DBM $M$ with entries in $\mathrm{Supp}_\val$ and
any atomic DBM $M'$, all entries of $M\cap M'$ are contained in
$\mathrm{Supp}(\val)$ if $M\cap M'$ is consistent.

Let $M'=\{(\prec'_{i,j},m'_{i,j})\}$ be an atomic DBM whose single
non-trivial constraint is $M'_{p,q}$ for some indices $p,q$ (i.e., all
other entries are $(<,\infty)$).  Then $m'_{p,q}\in\mathbb{Z}$ by
definition of  atomic DBMs.
Recall that the DBM $M''= M \cap M'$
is given by
\[ M''_{i,j} = \min(M_{i,j},M_{i,p}+M'_{p,q}+M_{q,j}) \]
 for all indices $i,j$.  Suppose $M''=M\cap M'$ is consistent and
recall by Proposition~\ref{prop:tight} that $M$ is tight.

Fix indices $0 \leq i,j \leq n$.  We show that
$m''_{i,j} \in \mathrm{Supp}_\val$.  
If $M''_{i,j} = M_{i,j}$ then $m''_{i,j} \in \mathrm{Supp}_\val$ by
	the induction hypothesis.  So we may suppose that 
\begin{gather}
M''_{i,j} = M_{i,p}+M'_{p,q}+M_{q,j} < M_{i,j}
\label{eq:sum}
\end{gather}
By the induction hypothesis, $m_{i,p},m_{q,j} \in \mathrm{Supp}_\val$.
From (\ref{eq:sum}) we must have $m_{i,p},m_{q,j} <\infty$.  We now
consider three cases.

\begin{enumerate}
\item Suppose that $m_{i,p} \in \mathbb{Z}$.
Then  $m''_{i,j}$ has the form $d+m_{q,j}$ for some integer $d$, and hence 
$m''_{i,j} \in \mathrm{Supp}_\val$ by the induction hypothesis.
\item Suppose that $m_{q,j} \in \mathbb{Z}$.
Then  $m''_{i,j}$ has the form $d+m_{i,p}$ for some integer $d$, and hence 
$m''_{i,j} \in \mathrm{Supp}_\val$ by the induction hypothesis.
\item The final case is that $m_{i,p},m_{q,j} \not\in \mathbb{Z}_\infty$. Then
\begin{eqnarray*}
M_{i,p}+M'_{p,q}+M_{q,j}&=& M_{i,n}+M_{n,p}+M'_{p,q}+M_{q,n}+M_{n,j} \quad\mbox{ ($M$ tight)}\\
&\geq & M_{i,n}+M''_{n,p}+M''_{p,q}+M''_{q,n}+M_{n,j} \quad\mbox{ ($M,M' \geq M''$ pointwise)}\\
&\geq & M_{i,n} + M''_{n,n} + M_{n,j}\quad\mbox{ ($M''$ canonical)}\\
& \geq  & M_{i,n}+M_{n,j} \quad\mbox{ ($M''$ consistent)}\\
& \geq & M_{i,j} \quad\mbox{ ($M$ canonical).}
\end{eqnarray*}
But this contradicts (\ref{eq:sum}) and so this case cannot hold.
\end{enumerate}

%
%
%
%
%

\end{proof}

\subsection{Proof of Propositions~\ref{prop:monotone} and Lemma~\ref{lem:cut-down}}
\label{append:monotone}
Let $\Sigma = \{a_1,\ldots,a_n\}$ be a finite alphabet.  Define a
function $\pi : \Sigma^* \rightarrow \mathbb{N}^n$ such that
$\pi(w)_i$ is the number of occurrences of letter $a_i$ in $w$ for
$i=1,\ldots,n$.  The image of a language $L\subseteq\Sigma^*$ under
$\pi$ is called the \emph{Parikh image} (or \emph{commutative image})
of $L$.  It is well known that the Parikh image of any regular
language (indeed any context-free language) is definable in Presburger
arithmetic. In particular, the Parikh image of the language of an NFA
over a unary alphabet is a union of arithmetic progressions.  Chrobak
and Martinez~\cite{Chrobak86,Martinez02} show that the Parikh image of
the language of an $n$-state NFA $\mathcal{A}$ over a unary alphabet
comprises $O(n^2)$ many arithmetic progressions which can be
explicitly computed from $\mathcal{A}$ in polynomial time.

Consider a monotonic counter machine $\mathcal{C}=(S,C,\Delta)$.  Let $N$ be
the maximum constant appearing in a transition guard.  Define the set
$\mathcal{R}_N$ of \emph{regions} to be $\mathcal{R}_N=
\{0,\ldots,N\} \cup \{\infty\}$.  A counter valuation $\upsilon \in
\mathbb{N}^n$ defines a tuple $Reg(\upsilon) \in \mathcal{R}_N^n$ by
\[ Reg(\upsilon)_i = \left\{ \begin{array}{ll} \upsilon_i & \mbox{ if $\upsilon_i \leq N$}\\
\infty & \mbox{ otherwise} \end{array}\right . \] 
Intuitively $\infty$
represents any counter value strictly greater than $N$. 
The satisfaction relation~$\models$ between regions and guards
is defined in the obvious way.
Below we define a
finite automaton $[\mathcal{C}]$ that simulates $\mathcal{C}$.

The alphabet of $[\mathcal{C}]$ is
$\Sigma = \{ \mathit{inc}_1,\ldots,\mathit{inc}_n\}$.  Intuitively
$[\mathcal{C}]$ performs an $\mathit{inc}_i$-transition when
simulating an increment on counter $c_i$.  A state of $[\mathcal{C}]$
is a tuple $\tuple{s,R,\lambda}$, where $s \in S$,
$R \in \mathcal{R}_N^n$ is a region of
$\mathcal{C}$, and $\lambda \subseteq C$.    With a
configuration $\tuple{s,\upsilon}$ in a run $\rho$ of $\mathcal{C}$ we
associate a state $\tuple{s,Reg(\upsilon),\lambda}$ of $[\mathcal{C}]$.
Intuitively, $\lambda$ represents the set of counters that will be
reset along the suffix of the run $\rho$ starting from
$\tuple{s,\upsilon}$.

The transition relation of $[\mathcal{C}]$ is defined as follows:
\begin{itemize}
\item
For each edge $\tuple{s,\varphi,\mathit{reset}(c_i),s'} \in \Delta$ we add a
transition
$\tuple{s,R,\lambda} \stackrel{\varepsilon}{\longrightarrow}
\tuple{s',R',\lambda'}$ if $R\models \varphi$,
$R'_i = 0$, $R'_j=R_j$ for $j\neq i$, and
$\lambda'\cup\{c_i\} = \lambda$.
\item
For each edge $\tuple{s,\varphi,\mathit{nop},s'} \in \Delta$ we add a
transition
$\tuple{s,R,\lambda} \stackrel{\varepsilon}{\longrightarrow}
\tuple{s',R,\lambda}$ if $R\models \varphi$.
\item
For each edge $\tuple{s,\varphi,\mathit{inc}(c_i),s'} \in \Delta$ we add
a transition
$\tuple{s,R,\lambda} \stackrel{\sigma}{\longrightarrow}
\tuple{s',R',\lambda}$ if $R\models \varphi$,
$R'_i=R_i+1$, and $R'_j=R_j$ for
$j\neq i$.  The label $\sigma$ is
$\mathrm{inc}_i$ if $c_i\not\in\lambda$ and otherwise $\sigma$ is
$\varepsilon$.
\end{itemize}

By construction of $[\mathcal{C}]$, there is a run of $\mathcal{C}$
from $\tuple{s,\upsilon}$ to $\tuple{s',\upsilon'}$ along which the
collection of counters that are reset is $\lambda=\{c_1,\ldots,c_m\}$
only if there is a run of $[\mathcal{C}]$ from
$\tuple{s,Reg(\upsilon),\lambda}$ to $\tuple{s',Reg(\upsilon'),\emptyset}$.
If $w\in\Sigma^*$ is the word read along such a run then we have
\begin{equation}
\begin{aligned}
\upsilon'_i &= \pi(w)_i \quad i=1,\ldots,m\\
\upsilon'_i - \upsilon_i & = \pi(w)_i \quad i=m+1,\ldots,n \, .
\end{aligned}
\label{eq:upsilon}
\end{equation}

Fix states $\tuple{s,R,\lambda}$ and $\tuple{s',R',\emptyset}$ of
$[\mathcal{C}]$.  Let
$L_{\tuple{s,R,\lambda},\tuple{s',R',\emptyset}}$ be the set of words
$w$ on which $[\mathcal{C}]$ has a run from $\tuple{s,R,\lambda}$ to
$\tuple{s',R',\emptyset}$.  Then the Parikh image
$\pi(L_{\tuple{s,R,\lambda},\tuple{s',R',\emptyset}})$ is expressible
by a formula $\psi(z_1,\ldots,z_n)$ of Presburger arithmetic.

Returning to the counter machine $\mathcal{C}$, we wish to express the
reachability relation of $\mathcal{C}$ between two controls states $s$
and $s'$.  The idea is that for each initial counter valuation
$\upsilon$ and each run of $\mathcal{C}$ from $\tuple{s,\upsilon}$ to
$s'$, we need to specify the total number of increments for each
counter that is never reset along the run and the total number of
increments since the last reset for all other counters.  With this in
mind, using Equation (\ref{eq:upsilon}), the
$\mathcal{L}_{\mathbb{Z}}$-formula
\[ \varphi(\upsilon,\upsilon') \defequals (Reg(\upsilon)=R) \wedge
  \psi(\upsilon'_1,\ldots,\upsilon'_m,\upsilon'_{m+1}-\upsilon_{m+1},\ldots,
  \upsilon'_n-\upsilon_n) \] describes the subset of the reachability
relation arising from the runs of $\mathcal{C}$ whose projection on
$[\mathcal{C}]$ goes from state $\tuple{s,R,\lambda}$ to
$\tuple{s',R',\emptyset}$, for $\lambda=\{c_1,\ldots,c_m\}$. The
reachability relation of $\mathcal{C}$ can clearly be described as a
finite disjunction of such formulas.  This concludes the proof of
Proposition~\ref{prop:monotone}.

The following specialisation of Proposition~\ref{prop:monotone} is
used in the proof of Lemma~\ref{lem:cut-down}.

\begin{proposition}\label{prop:reference}
  Let $\mathcal{C}$ be a monotonic counter machine with~$n$ counters
and with $N$ the maximum integer constant appearing in a transition guard.
  Given states $s,s'$ of $\mathcal{C}$ and 
  $R,R'\in \mathcal{R}_N^n$, the set
  \begin{align*}
 \{ (u,u') \in \mathbb{N}^2 :& \;  \exists \, \tuple{s,\upsilon} \mathrel{\longrightarrow^*}
    \tuple{s',\upsilon'} \mbox{ s.t. } \\ 
& Reg(\upsilon)=R, Reg(\upsilon')= R',
  \upsilon_n=u, \upsilon'_n=u'\} 
\end{align*} 
is definable by a quantifier-free formula of
$\mathcal{L}^*_{\mathbb{R},\mathbb{Z}}$ (involving only integer
terms) that is computable in time polynomial in (the  largest (absolute) constant of)
$N$ and the number of states and counters of $\mathcal{C}$.
\end{proposition}
\begin{proof}
We start by defining an NFA $\mathcal{B}$, over a singleton alphabet
$\{ \mathit{inc}_n \}$.  Automaton $\mathcal{B}$ can be seen as a
``sub-automaton'' of the NFA $[\mathcal{C}]$ from the proof of
Proposition~\ref{prop:monotone}.  Specifically the states of
$\mathcal{B}$ are those states $\tuple{s'',R'',\lambda}$ of
$[\mathcal{C}]$ such that either $\lambda=C$ or $\lambda =
C\setminus\{c_n\}$.  (This last condition means that all increments of
counters other than $c_n$ are represented in $\mathcal{B}$ by
$\varepsilon$-transitions.)  For the fixed states and regions
$s,R,s',R'$, as in the statement of the proposition, the initial
states of $\mathcal{B}$ are those of the form $\tuple{s,R,\lambda}$,
where $\lambda=C$ or $\lambda = C\setminus\{c_n\}$, and the accepting
states those of the form $\tuple{s',R',C\setminus\{c_n\}}$.

Then the Parikh image of the language of $\mathcal{B}$ is equal to
\begin{gather}\{ (\upsilon_n,\upsilon'_n) \in \mathbb{N}^2 : Reg(\upsilon)=R, Reg(\upsilon')= R',
  \tuple{s,\upsilon} \mathrel{(\longrightarrow)^*}
  \tuple{s',\upsilon'} \} \, . 
\label{eq:semialg}
\end{gather}
We can now appeal to the above-mentioned result of Chrobak and
Martinez~\cite{Chrobak86,Martinez02} to get that the set
(\ref{eq:semialg}) is definable by a quantifier-free formula of
Presburger arithmetic that is computable in time polynomial in the
number of states of $\mathcal{B}$, that is, polynomial in 
the largest (absolute) constant of $N$ and the number of states and counters of
$\mathcal{C}$.
\end{proof}

The proof of Lemma~\ref{lem:cut-down} is exactly the same as the proof
Theorem~\ref{theorem_main_ta_formula}, except that we replace the use of Proposition~\ref{prop:monotone}
by Proposition~\ref{prop:reference}, so as to obtain a quantifier-free
formula in $\mathcal{L}^*_{\mathbb{R},\mathbb{Z}}$.

\subsection{Proof of Proposition~\ref{prop:sound}}
\label{app:sound}
We first give the ``soundness'' direction of the proof, that is, from
runs of the counter machine $\mathcal{C}_{\tuple{\loc,\nu}}$ 
to runs of $\A$.

Suppose that 
\[ \tuple{\tuple{\loc_0,M_0},\upsilon^{(0)}} \longrightarrow
	\tuple{\tuple{\loc_1,M_1},\upsilon^{(1)}} \longrightarrow
   \ldots \longrightarrow 
\tuple{\tuple{\loc_k,M_k},\upsilon^{(k)}} \]
is a run of $\mathcal{C}_{\tuple{\loc,\nu}}$ with $\loc_0=\loc$, $\upsilon^0=\floor{\val}$ and
$\semantics{M_0} = \{ \fract(\val) \}$.
Given any valuation $\nu^{(k)} \in \semantics{M_k}$,
we construct a sequence of valuations $\nu^{(0)},\ldots,\nu^{(k-1)}$,
with $\nu^{(j)} \in \semantics{M_j}$ for $j=0,\ldots,k-1$, such that
\[ \tuple{\loc_0,\upsilon^{(0)}+\nu^{(0)}} \Longrightarrow
   \tuple{\loc_1,\upsilon^{(1)}+\nu^{(1)}} \Longrightarrow
   \ldots \Longrightarrow
\tuple{\loc_k,\upsilon^{(k)}+\nu^{(k)}} \]
is a run of $\A$.  Note that then we must have $\val^{(0)}=\fract(\val) $.

The construction of $\nu^{(j)}$ is by backward induction on $j$.  The
base step, valuation $\nu^{(k)}$, is given.  The induction step
divides into three cases according to the nature of the transition
$\tuple{\tuple{\loc_{j-1},M_{j-1}},\upsilon^{(j-1)}}
\longrightarrow \tuple{\tuple{\loc_{j},M_{j}},\upsilon^{(j)}}$.  (Recall the classification
of transitions in the definition of $\mathcal{C}_{\tuple{\loc,\nu}}$.)
\begin{itemize}
\item $\tuple{\tuple{\loc_{j-1},M_{j-1}},\upsilon^{(j-1)}}                                                                         \longrightarrow \tuple{\tuple{\loc_{j},M_{j}},\upsilon^{(j)}}$ is a delay transition.
Then we have $M_j =
\overrightarrow{M_{j-1}} \cap [0,1]^n$, $\loc_j = \loc_{j-1}$, and
$\upsilon^{(j)}=\upsilon^{(j-1)}$.  Thus we can pick $\nu^{(j-1)} \in
\semantics{M_{j-1}}$ such that $\nu^{(j)}=\nu^{(j-1)}+d$ for some
$d\geq 0$.
Thus there is a delay transition
\[ \tuple{\loc_{j-1},\upsilon^{(j-1)}+\nu^{(j-1)}} \stackrel{d}{\Longrightarrow}
   \tuple{\loc_j,\upsilon^{(j)}+\nu^{(j)}} \] in $\A$.
\item
	$\tuple{\tuple{\loc_{j-1},M_{j-1}},\upsilon^{(j-1)}}                                                                         
	\longrightarrow \tuple{\tuple{\loc_{j},M_{j}},\upsilon^{(j)}}$ is
a wrapping transition.  Then we have $M_j = (M_{j-1} \cap
(x_i=1))[x_i \leftarrow 0]$ for some index $i$.  Thus we can pick
$\nu^{(j-1)} \in \semantics{M_{j-1} \cap (x_i=1)}$ such that
$\nu^{(j)}=\nu^{(j-1)}[x_i\leftarrow 0]$.  In this case we have
\[ \tuple{\loc_{j-1},\upsilon^{(j-1)}+\nu^{(j-1)}} =
   \tuple{\loc_j,\upsilon^{(j)}+\nu^{(j)}} \, . \]
\item
	$\tuple{\tuple{\loc_{j-1},M_{j-1}},\upsilon^{(j-1)}}                                                                         
	\longrightarrow \tuple{\tuple{\loc_{j},M_{j}},\upsilon^{(j)}}$ 
is a discrete transition.  Let the 
corresponding edge of $\A$ be $\tuple{\loc_{j-1},\varphi,\{x_i\},\loc_j}$.  Then we have $M_j =
(M_{j-1} \cap \varphi_{\mathsf{frac}})[x_i \leftarrow 0]$.  Thus we
may pick
$\nu^{(j-1)} \in \semantics{M_{j-1} \cap \varphi_{\mathsf{frac}}}$ such
that $\nu^{(j-1)}[x_i \leftarrow 0] = \nu^{(j)}$.  Since
$\upsilon^{(j-1)} \models \varphi_{\mathsf{int}}$ we have that
$\upsilon^{(j-1)}+\nu^{(j-1)} \models \varphi$.  Thus there is a discrete
transition
\[ \tuple{\loc_{j-1},\upsilon^{(j-1)}+\nu^{(j-1)}} \stackrel{0}{\Longrightarrow}
   \tuple{\loc_j,\upsilon^{(j)}+\nu^{(j)}} \] in $\A$.
\end{itemize}

We now give the ``completeness'' direction of the proof: from runs of the timed automaton 
$\A$ to runs of the counter machine $\mathcal{C}_{\tuple{\loc,\nu}}$.

Suppose that we have a run 
\[ \tuple{\loc_0,\nu^{(0)}} \stackrel{d_1}{\Longrightarrow} 
   \tuple{\loc_1,\nu^{(1)}} \stackrel{d_2}{\Longrightarrow} 
\ldots \stackrel{d_k}{\Longrightarrow} \tuple{\loc_k,\nu^{(k)}} \]
  of $\A$, where $\tuple{\loc_0,\nu^{(0)}}=\tuple{\loc,\nu}$.  We can
  transform such a run, while keeping the same initial and final
  configurations, by decomposing each delay step into a sequence of
  shorter delays, so that for all $0 \leq j \leq k-1$ and all
  $x\in\clocks$ the open interval $(\nu^{(j)}(x),\nu^{(j+1)}(x))$
  contains no integer.  In other words, we break a delay step at any
  point at which some clock crosses an integer boundary.  We can now
  obtain a corresponding run of $\mathcal{C}_{\tuple{\loc,\nu}}$ that
  starts from state 
  $\tuple{\tuple{\loc_0,M_0},\upsilon^{(0)}}$, where
$\semantics{M_0} = \{ \mathrm{frac}(\nu) \}$ 
and $\upsilon^{(0)} = \lfloor \nu^{(0)} \rfloor$,
and ends in state $\tuple{\tuple{\loc_k,M_k},\upsilon^{(k)}}$ such that
$\nu^{(k)} \in \upsilon^{(k)}+\semantics{M_k}$.

We build such a run of $\mathcal{C}_{\tuple{\loc,\nu}}$ by forward
induction.  
In particular, we construct a sequence of intermediate
states $\tuple{\tuple{\loc_i,M_i},\upsilon^{(i)}}$, $0 \leq i \leq k$, such
that $\nu^{(i)} \in \upsilon^{(i)} + \semantics{M_i}$ for each such
$i$.  Each discrete transition of $\A$ is simulated by a discrete
transition of $\mathcal{C}_{\tuple{\loc,\nu}}$.  A delay transition of
$\A$ that ends with set of clocks
$\lambda \subseteq \{x_1,\ldots,x_n\}$ being integer valued is
simulated by a delay transition of $\mathcal{C}_{\tuple{\loc,\nu}}$,
followed by wrapping transitions for all counters $c_i$ for which
$x_i \in \lambda$.

\subsection{Proof of Theorem~\ref{theorem_main_ta_formula}}
\label{TimeBound}
Let $\A=\tuple{\locs,\clocks,\edges}$ be a timed automaton with
maximum clock constant $N$.  We first transform $\A$ so that all
guards are conjunctions of atoms of the type appearing in
Figure~\ref{fig:decompose}.  This transformation may lead to an
exponential blow-up in the number of edges.  In any case, it can be
accomplished in time at most
$2^{\mathrm{poly}(n)} \cdot \mathrm{poly}(L)$.

Let $\tau$ be an $n$-type.  Following
Corollary~\ref{corl:well-supported} we have observed that
$|\closure(\mathcal{M}_{\tau})| \leq 2^{\mathrm{poly}(n)}$.  It
follows that for a location $\loc \in L$ and $n$-type $\tau$, the
monotonic counter automaton $\mathcal{C}_{\tuple{\loc,\tau}}$ can be
computed in time at most
$2^{\mathrm{poly}(n)} \cdot \mathrm{poly}(|L|)$.

Applying Proposition~\ref{prop:monotone}, we get that the
formula $\chi^\tau_{\loc,\loc'}$ can be computed in time at most
$2^{\mathrm{poly}(n)} \cdot \mathrm{poly}(|L|,N^n)$.  Furthermore,
given $\tau$, the formula $\alpha^\tau$ can be computed in time
$\mathrm{poly}(n)$.

Finally, the number of disjuncts in (\ref{eq-formula}), i.e., the
number of different $n$-types when restricting to formulas $t\leq t'$
for $t,t' \in \mathcal{DT}_{\mathbb{R}}(n)$, is bounded by
$2^{\mathrm{poly}(n)}$.

Putting everything together, the formula $\varphi_{\loc,\loc'}$ can be
computed in time at most $2^{\mathrm{poly}(n)} \cdot
\mathrm{poly}(|L|,N^n)$, that is, exponential in the size of the
original timed automaton $\A$.


\subsection{Symbolic Counter Machines}\label{symbolicautmatafigure}

In this section we illustrate Figure~\ref{fig:symbolicAutomata2} used in Example~\ref{example-family}.

\begin{figure*}[t]
\centering
\begin{tikzpicture}[>=latex',shorten >=1pt,node distance=1.9cm,on grid,auto,
state/.style={
           rectangle,
           draw=black, 
           minimum height=1.5cm,
					minimum width=1.7cm,
           inner sep=1pt,
           text centered,
           },
roundnode/.style={circle, draw,minimum size=1.2mm},]

\node [state,label={[label distance=.3cm]-90:counter machine $\mathcal{C}_{\tuple{\loc_0,\tau_2}}$}] (l0z0) at(0,0) {\scalebox{.8}{
	\begin{tabular}[t]{l}  \\ 
$\vect{(\leq,0)&  (\leq,-r_1)&  (\leq,-r_2)\\
					 (\leq, r_1)&     (\leq,0)&   (\leq,r_1-r_2)\\
(\leq,r_2)&  (\leq,r_2-r_1)&     (\leq,0)}$\end{tabular}}};
\node [draw=none] (dum) [above=.45 of l0z0] {{\small $\tuple{\ell_0,\M'_0}$}};

\node [state] (l0z1) [right=5cm of l0z0] {\scalebox{.8}{
	\begin{tabular}[t]{l}  \\ 
$\vect{(\leq,0)&  (\leq,-r_1)&  (\leq,-r_2)\\
					 (\leq,r_1 - r_2 + 1)&     (\leq,0)&   (\leq,r_1-r_2)\\
(\leq,1)&  (\leq,r_2-r_1)&     (\leq,0)}$\end{tabular}}};
\node [draw=none] (dum) [above=.45 of l0z1] {$\tuple{\ell_0, \M'_1}$};

\node [state] (l1z2) [below=2.5cm of l0z1] {\scalebox{.8}{
	\begin{tabular}[t]{l}  \\ 
$\vect{(\leq,0)&  (\leq,0)&  (\leq,-r_2)\\
					 (\leq,0)&     (\leq,0)&   (\leq,-r_2)\\
					 (\leq,1)&  (\leq ,1)&     (\leq,0)}$\end{tabular}}};
\node [draw=none] (dum) [above=.45 of l1z2] {{\small $\tuple{\ell_1,\M'_2}$}};

\node [state] (l1z3) [right=5cm of l1z2]  {\scalebox{.8}{
	\begin{tabular}[t]{l}  \\ 
$\vect{(\leq,0)&  (<,0)&  (\leq,-r_2)\\
					 (\leq,1-r_2)&     (\leq,0)&   (\leq,-r_2)\\
					 (\leq,1)&  (\leq,1)&     (\leq,0)}$\end{tabular}}};
\node [draw=none] (dum) [above=.45 of l1z3] {{\small $\tuple{\ell_1,\M'_3}$}};

\node [state] (l1z4) [below=2.2cm of l1z3] {\scalebox{.8}{
	\begin{tabular}[t]{l}  \\ 
$\vect{(\leq,0)&  (\leq,0)&  (\leq,0)\\
					 (\leq,1-r_2)&     (\leq,0)&   (\leq,1-r_2)\\
					 (\leq,0)&  (\leq,0)&     (\leq,0)}$\end{tabular}}};
\node [draw=none] (dum) [above=.45 of l1z4] {{\small $\tuple{\ell_1,\M'_4}$}};
\node [state] (l1z5) [below=2.2cm of l1z2]{\scalebox{.8}{
	\begin{tabular}[t]{l}  \\ 
$\vect{(\leq,0)&  (\leq,0)&  (\leq,0)\\
					 (\leq,1)&     (\leq,0)&   (\leq,1-r_2)\\
					 (\leq,1)&  (\leq,0)&     (\leq,0)}$\end{tabular}}};
\node [draw=none] (dum) [above=.45 of l1z5] {$\tuple{\ell_1,\M'_5}$};
\path[->] (l0z0) edge node [above,midway] {\scriptsize{$\mathit{nop}$}}  (l0z1);
\path[->] (l0z1) edge node [near start,right] {\scriptsize{$\mathit{reset}(c_1)$}} node [right,near end] {\scriptsize{$c_1=0$}}  (l1z2);
\path[->] (l1z2) edge node [above,midway] {\scriptsize{$\mathit{nop}$}} (l1z3);
\path[->] (l1z3) edge node [left,midway] {\scriptsize{$\mathit{inc}(c_2)$}}  (l1z4);
\path[->] (l1z4) edge node [above,midway] {\scriptsize{$\mathit{nop}$}} (l1z5);
\path[->] (l1z5) edge node [right,midway] {\scriptsize{$\mathit{inc}(c_1)$}}   (l1z2);

\end{tikzpicture}
\caption{The counter machine~$\mathcal{C}_{\tuple{\loc,\tau_2}}$ constructed from the timed automaton in Figure~\ref{fig:reachAutomata},
 where $\tau_2$ is the type of the valuation~$\val$ with $\val_1=0$ and $\val_2=.2$.  
}
\label{fig:symbolicAutomata2}
\end{figure*}

\subsection{Proof of Theorem~\ref{main-theo-NEXPTIME}}
\label{NEXPTIME-hard}

This section we continue the argument of
 Section~\ref{sec:nexptimehard} showing that model checking parametric
 timed reachablity logic is NEXPTIME-hard.

It remains to explain how from linear bounded automaton $\mathcal{B}$
one can define a timed automaton $\A$ whose configuration graph embeds
the configuration graph of $\mathcal{B}$.  The construction is adapted
from the PSPACE-hardness proof for reachability in timed
automata~\cite{AlurD94}.  We assume that $\mathcal{B}$ uses a binary
input alphabet and a fixed tape length of $k$.  The main idea is as
follows: $\A$ uses $2k+1$ clocks: one clock $y_i$ and $z_i$ for each
tape cell $i$, and one extra clock $x$.  The clocks $y_i$ and $z_i$,
respectively, are used to encode the current tape content and the
position of the pointer of $\mathcal{B}$, respectively.  The clock $x$
is an auxiliary clock that helps to encode this information correctly
into the other clocks. Technically, $x$ is used to measure out cycles
of two time units, i.e., $x$ is reset to $0$ whenever it reaches $2$.
The construction is such that the values of $y_i$ and $z_i$ obey the
following policy: whenever $x$ takes value $0$, $y_i$ takes value $1$
($0$, respectively) if there is a $1$ ($0$, respectively) in the
$i$-th cell of the tape; and $z_i$ takes value $1$ if the position of
the pointer is the $i$-th cell, otherwise, $z_i$ takes value $0$.  We
can set these bits appropriately by resetting clocks $y_i$ and $z_i$
either when $x=1$ or $x=2$, and we can preserve the values of a clock
$y_i$ or $z_i$ between successive cycles by resetting it when it
reaches value $2$, see below for more details. Using this idea, $\A$
can be defined such that it only takes transitions at integer times
and such that a configuration of $\A$ after $2t$ time steps encodes a
configuration of $\mathcal{B}$ after $t$ computation steps for each
$t\in\mathbb{N}$.

More formally, the set of locations of 
$\A$ contains one copy location $q$ for each state $q$ of $\mathcal{B}$, plus some additional auxiliary locations, one of which being an initial location $\ell_0$.  
In the initialization phase, we encode the initial configuration $(q_0,\sigma_1)\sigma_2\dots\sigma_k$ of $\mathcal{B}$, where $q_0$ is the initial state of $\mathcal{B}$, and $\sigma_i\in\{0,1\}$. For this, we define a transition from $\ell_0$ to $q_0$, with guard $x=1$, and resetting $x$, $z_2, \dots, z_k$, and we further reset clock $y_i$ iff $\sigma_i=0$.  
One can easily observe that if $\A$ reaches $q_0$ with clock value $x=0$, then $z_1=1$, 
and $y_i=1$ iff the $i$-th cell contains a $1$, while all other clocks have value $0$. This correctly encodes the initial configuration of $\mathcal{B}$. 
We now proceed with the simulation phase. 
From locations $q$ that correspond to states of $\mathcal{B}$, we simulate the computation steps from $\mathcal{B}$. 
Assume, for instance, that the transition relation of $\mathcal{B}$ contains the tuple $(q,0,q',0,R)$, i.e., when reading letter $0$ in state $q$, $\mathcal{B}$ goes to state $q'$, leaves the symbol on the tape as it is, and moves the pointer one position to the right. 
According to the encoding described above, 
this means that if $\A$ reaches $q$ with $x=0$, 
we need to test whether $y_i=0$ for the unique $1\leq i\leq k$ such that $z_i=1$, and whether $i<k$ (because we want to move the position of the pointer one cell to the right). If this is the case, $\A$ should go to location $q'$ and the bit of $z_i$ should be reset to $z_{i+1}$. We thus define for every $1\le i<k$ a transition as shown in the following, where the loops in the auxiliary location in the middle are defined for every $1\le j\le k$. 
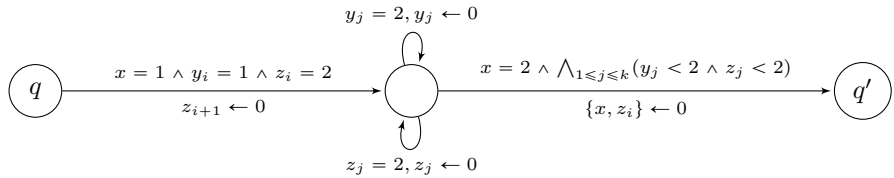
\begin{figure}[h!]
\begin{center}
\begin{tikzpicture}[>=latex',shorten >=1pt,node distance=1.9cm,on grid,auto,
roundnode/.style={circle, draw,minimum size=.7cm}]

\node [roundnode] (li) at(0,0) {$q$};
\node [roundnode] (l0) [right=5cm of li] {$\, $};
\node [roundnode](l1)  [right=6cm of l0] {$q'$};

\path[->] (li) edge  node [above,midway] {\scriptsize{$x=1\wedge y_i=1\wedge z_i=2$}} node[below, midway]{ \scriptsize{$z_{i+1}\leftarrow 0$}} (l0);

\path[->] (l0) edge  node [above,midway] {\scriptsize{$x=2 \wedge\bigwedge_{1\le j\le k} (y_j<2\wedge z_j<2)$}} node[below, midway]{ \scriptsize{$\{x,z_i\}\leftarrow 0$}} (l1);

\path[->] (l0) edge  [loop above] node [above,midway] {\scriptsize{$y_j=2, y_j\leftarrow 0$}}  (l0);

\path[->] (l0) edge  [loop below] node [below,midway] {\scriptsize{$z_j=2, z_j\leftarrow 0$}}  (l0);

\end{tikzpicture}
\end{center}
\end{figure}

Transitions of $\mathcal{B}$ of other forms can be simulated in a similar way. 
We finally augment $\A$ with a
label function $\lbl$ that assigns $p$ to a location $q$ iff
$q$ is an accepting state of $\mathcal{B}$. 

This finishes the proof for NEXPTIME-hardness.

\end{document}